\documentclass{article}
\usepackage{fullpage}
\usepackage[ruled]{algorithm2e} 
\usepackage{natbib} 
\usepackage{authblk}
\usepackage[bottom]{footmisc}

\usepackage[english]{babel}
\usepackage[utf8x]{inputenc}
\usepackage{amsfonts}
\usepackage{amsmath}
\usepackage{amssymb}
\usepackage{amsthm}
\usepackage{bbm}
\usepackage{bm}
\usepackage{booktabs}
\usepackage{cleveref}
\usepackage{enumitem}
\usepackage{graphicx}
\usepackage{mathtools}
\usepackage{multirow}
\usepackage{nicefrac}
\usepackage{subcaption}
\usepackage{thm-restate}
\usepackage[colorinlistoftodos]{todonotes}
\usepackage{wrapfig}
\usepackage{xspace}

\newtheorem{theorem}{Theorem}

\newtheorem{definition}[theorem]{Definition}
\newtheorem{lemma}[theorem]{Lemma}
\newtheorem{corollary}[theorem]{Corollary}

\newtheorem{problem}{Open Problem}

\DeclareMathOperator*{\argmax}{argmax}

\newcommand{\jonnote}[1]{}
\newcommand{\joshnote}[1]{}
\newcommand{\ggnote}[1]{}

\newcommand{\indicator}[1]{\mathcal{I}\left(#1\right)}
\newcommand{\D}{\mathcal{D}}
\newcommand{\E}{\mathbb{E}}

\newcommand{\N}{\mathbb{N}}

\def \Welf  {{\sf Welf}}

\DeclarePairedDelimiter{\set}{ \{ }{ \} }

\newcommand{\dist}[3]{d_{#3}\left(#1, #2\right)}

\newcommand{\poly}{\mathrm{poly}}

\newcommand{\Derivative}[1]{\frac{d}{d#1}}

\newcommand{\DSet}{\textnormal{Dominating-Set}\xspace}

  \newcommand{\cost}[1]{c_{#1}}
  \newcommand{\costvector}{\bm{c}}
  \newcommand{\Cost}[1]{c_{#1}}
  \newcommand{\CostVector}{\bm{c}}
  \newcommand{\TypedCostVector}[1]{\bm{c}^{(#1)}}
  
  \newcommand{\Forecast}[2]{F_{#1, #2}}
  \newcommand{\TypedForecast}[3]{F^{(#1)}_{#2, #3}}
  \newcommand{\ForecastMatrix}{F}
  \newcommand{\TypedForecastMatrix}[1]{F^{(#1)}}
  \newcommand{\ForecastTensor}{\bm{F}}
  
  \newcommand{\reward}[1]{r_{#1}}
  \newcommand{\rewardvector}{\bm{r}}
  \newcommand{\Reward}[1]{r_{#1}}
  \newcommand{\RewardVector}{\bm{r}}
  
  \newcommand{\TypedExpectedReward}[2]{R^{(#1)}_{#2}}
  \newcommand{\ExpectedReward}[1]{R_{#1}}
  
  \newcommand{\PrincipalAgentProblem}{\left(\costvector, \ForecastTensor, \rewardvector\right)}
  \newcommand{\PrincipalAgentProblemPrime}{\left(\costvector', \ForecastTensor', \rewardvector'\right)}
  
  \newcommand{\Contract}[1]{x_{#1}}
  \newcommand{\TypedContract}[2]{x^{(#1)}_{#2}}
  \newcommand{\ContractVector}{\bm{x}}
  \newcommand{\TypedContractVector}[1]{\bm{x}^{(#1)}}

  \newcommand{\ContractMatrix}{X}

  \newcommand{\TransferCoefficient}{\alpha}

  \newcommand{\OptimalAction}{i_*}
  \newcommand{\TypedOptimalAction}[1]{\OptimalAction^{(#1)}}

  \newcommand{\Profit}{\textsc{Profit}}
  \newcommand{\TypedProfit}[1]{\Profit^{(#1)}}
  \newcommand{\Utility}{U}

  \newcommand{\forecastprime}[3]{F^{\prime(#1)}_{#2, #3}}
  \newcommand{\contract}[1]{x_{#1}}

\newcommand{\Measure}{B}
\newcommand{\OWelfare}{\textsc{Welfare}}
\newcommand{\OTypeAware}{\textsc{Opt-TypeAware}}
\newcommand{\OMenu}{\textsc{Opt-Menu}}
\newcommand{\OSingle}{\textsc{Opt-Single}}
\newcommand{\OLinear}{\textsc{Opt-Linear}}

\newcommand{\StandardBasisVector}[1]{\vec{e}_{#1}}

\newcommand{\ApxHard}{\textsc{APX}-Hard\xspace}
\newcommand{\NPHard}{\textsc{NP}-Hard\xspace}

\newcommand{\dutting}{D{\"u}tting}

\newcommand{\Naively}{Na{\"i}vely}

\newcommand{\actionmapping}{(i_1, i_2, \ldots, i_T)}
\newcommand{\optactionmapping}{(i^*_1, i^*_2, \ldots, i^*_T)}

\newcommand{\neighbor}[2]{\delta_{#1}\left(#2\right)}

\title{Contracts under Moral Hazard and Adverse Selection}

\author[1]{Guru Guruganesh, Jon Schneider, Joshua Wang }
\affil[1]{Google Research \footnote{ \{gurug,jschnei,joshuawang\}@google.com}}

\begin{document}

\maketitle

\begin{abstract}
In the classical principal-agent problem, a principal must design a contract to incentivize an agent to perform an action on behalf of the principal. We study the classical principal-agent problem in a setting where the agent can be of one of several types (affecting the outcome of actions they might take). This combines the contract theory phenomena of ``moral hazard'' (incomplete information about actions) with that of ``adverse selection'' (incomplete information about types). 

We examine this problem through the computational lens. We show that in this setting it is APX-hard to compute either the profit-maximizing single contract or the profit-maximizing menu of contracts (as opposed to in the absence of types, where one can efficiently compute the optimal contract). We then show that the performance of the best linear contract scales especially well in the number of types: if agent has $n$ available actions and $T$ possible types, the best linear contract achieves an $O(n\log T)$ approximation of the best possible profit. Finally, we apply our framework to prove tight worst-case approximation bounds between a variety of benchmarks of mechanisms for the principal.
\end{abstract}

\thispagestyle{empty}

\pagebreak

\clearpage
\setcounter{page}{1}

\section{Introduction}
\label{sec:intro}

Consider the classical principal-agent problem, wherein one self-interested economic actor (the ``agent'') is hired to take actions on behalf of another (the ``principal''). The problem arises because the agent has little to no stake in the resulting outcome, so the two enter into a contract designed to encourage the agent to take actions in the interest of the principal. Contract theory distinguishes between two fundamental information asymmetries to consider in the design of this contract. The first is \emph{moral hazard}, which occurs when it may be difficult or costly for the principal to monitor the agent's actions. In this case, the agent may take inappropriate actions, e.g. an insured party may begin to act more recklessly since they are shielded from risk. The second is \emph{adverse selection} (also known as screening), where the agent possesses some private information about their type (e.g. how costly it is for them to perform certain actions) before they decide whether to enter the contract. In this case, the agents that enter contract may be biased, e.g.~people with pre-existing conditions would prefer to purchase health insurance.

We study the confluence of these two asymmetries, moral hazard and adverse selection, under the computational lens. In particular, we consider a simple hidden-action principal-agent model~\cite{grossman1992analysis} but generalized to capture multiple possible agent types. In our generalized model, the agent may be one of $T$ different types. The agent has the choice between $n$ actions, each of a different cost to the agent (typically representing different effort levels the agent may exert). After the agent chooses an action, one of $m$ disjoint outcomes may occur; the probability of a particular outcome depends on both the agent's type and the chosen action. The principal observes this outcome, but does not observe the chosen action (creating a moral hazard issue) nor the type of the agent (creating a adverse selection issue).

The principal also receives a reward depending on the outcome that occurred. They wish to incentivize the agent to play actions that generate high rewards in expectation. To this end, the principal may offer the agent a \emph{contract}: a mapping from outcomes to transfer amounts with the promise that if a particular outcome occurs, the principal will transfer the corresponding amount to the agent. The principal wishes to maximize their total \emph{profit}, the rewards they receive less the amount they must transfer to the agent. We emphasize that this model is quite general and many applications do indeed have a principal subject to both information asymmetries. For example, consider a health insurance provider who can offer several different insurance policies (i.e. a menu of contracts), which incentivize different types of agents (e.g. those with and without pre-existing health conditions) to perform certain actions (e.g. to exercise, to get regular checkups). The principal must worry about both the moral hazard aspect of contracts while also accounting for adverse selection.


\subsection{Our Results}
\subsubsection*{Contracts without Types}

Before we proceed to our results, we give a brief sketch of analogous results for the principal-agent 
problem without types. Many of these results appear in \cite{dutting2019simple}, which was one of the first papers to 
study contracts from an algorithmic/computational perspective.

Without types, it is possible to efficiently find the optimal contract for a given principal-agent problem 
-- one way to do so is to, for each action, find the optimal contract that induces the agent to play this 
action by solving a linear program. However, the optimal contracts 
are often complex and unnatural: they can be non-monotone (better outcomes for the principal don't necessarily correspond to better payments for the agent), non-budget-balanced (the principal may pay more than their value for certain outcomes on those outcomes), and can be very sensitive to the parameters of the problem (in particular the forecast probabilities $\ForecastTensor$). 

Such contracts rarely occur in practice. Instead, people tend to use simpler, more-easily understood contracts. One of the simplest classes of contracts is the class of \emph{linear contracts} where the principal transfers a constant $\alpha$ fraction ($\alpha \in [0, 1]$) of their total reward to the agent. 

In \cite{dutting2019simple}, the authors also show that when there are $n$ actions the best linear contract achieves an $O(n)$-approximation to the best overall contract. In fact, their proof further shows that the profit of the best linear contract is an $O(n)$-approximation to the best profit achievable in the absence of moral hazard, when the principal can directly observe (and contract on) the agent's action (this is often referred  to in economics as the ``first-best'' outcome). In this case, the principal will persuade the agent to take the action $i$ (by paying them $c_i$ if and only if they play $i$) which maximizes the principal's expected reward less the agent's cost. In addition to being the optimal profit in the absence of moral hazard, this quantity can also be viewed as the maximum possible expected welfare -- we will henceforth refer to it as the ``optimal welfare''.  

\subsubsection*{Contracts with Types}

When the principal only knows that the agent is one of many possible types, the problem becomes more difficult.

To begin with, it is no longer necessarily true that the best mechanism for the principal consists of presenting the agent with a single contract. As is the case in traditional Bayesian mechanism design (which also has hidden types), the principal may wish to present the agent with a \emph{menu} of different contracts (each specifying how much the principal would transfer for each outcome). Each agent now begins by  choosing their most desirable contract from this menu, then picking the action which maximizes their utility as a result of this contract. 

In the previous section, it was possible to efficiently (in polynomial time in the inputs) compute the optimal contract. Now it is no longer clear how to efficiently compute the optimal menu for multiple types. In fact, it is not even clear how to efficiently compute the optimal \textit{single contract} for multiple types -- the same LP approach as before now requires solving exponentially many LPs, one for each assignment from types to actions. In~\Cref{sec:computation}, we show that this difficulty is inherent: we show that computation of both the optimal single contract and optimal menu are \ApxHard{}. That is, not only is it NP-hard to compute the optimal single contract (/optimal menu) exactly, there exists a constant $\epsilon > 0$ such that getting any $(1 - \epsilon)$-approximation to the optimal profit is \NPHard{}.

\begin{restatable}[APX-Hardness]{theorem}{apxhardness}
\label{thm:apx-hard}
  Given a principal-agent problem $\PrincipalAgentProblem$, it is \ApxHard to compute the optimal contract/menu. In other words, there exists a constant $\epsilon > 0$ such that it is \NPHard to compute a (incentive-compatible) menu of contracts $\ContractMatrix$ or single contract $\ContractVector$ whose value approximates the optimal expected profit by a factor of $(1 - \epsilon)$. 
\end{restatable}

We complement this result by showing that there exists an exponential time algorithm to compute the optimal contract. 
\begin{restatable}[Exponential-Time Algorithm]{lemma}{expalg}
\label{lem:algorithm}
There exists an algorithm to compute the optimal contract in time 
$$O(\min(\poly(n^T,m), \poly((n^2T)^m)))$$ 
\noindent
where $n$ is the number of actions, $m$ is the number of outcomes, and $T$ is the number of types. 
\end{restatable}
Our algorithm is tight in the following sense: if the number of types or number of outcomes is constant, then it is polynomial-time. However, the hardness instances that we generate only contain a constant number of actions and hence preclude an algorithm whose runtime is exponential only in the number of actions. 

The computational hardness of computing these optimal mechanisms make it even more important to understand the power of simple classes of contracts. For example, how does the profit of the best linear contract (which are efficiently computable) compare to that of the best single contract or best menu of contracts? One way to get such a bound is to look at the best linear contract for each type individually and take the best of these contracts: since each individual linear contract achieves an $O(n)$-approximation to the welfare for the corresponding type, and since there are $T$ types, this method will result in an $O(nT)$-approximation to the total welfare (and hence to the best single contract or best menu). In~\Cref{sec:linear}, we further explore the case where each agent type has the same cost values, and show that this structure means that the best linear contract actually scales as a $O(n \log T)$-approximation to the total welfare, and that this is tight.

\begin{theorem}[Linear contracts]\label{thm:linear_main}
For a given principal-agent problem with types, let $\OLinear$ denote the profit of the best linear contract and let $\OWelfare$ denote the optimal welfare. Then

$$\OLinear \ge \Omega\left(\frac{1}{n\log T}\right)\OWelfare.$$

Moreover, this is tight; for any $n$ and $T$ there exists an instance of the principal-agent problem where 

$$\OLinear \le O\left(\frac{1}{n\log T}\right)\OWelfare.$$
\end{theorem}

Our analysis pins down the pivotal problem parameter to be the number of unique ``cost gaps''. Consider each agent type separately, and sort the entries of each type's cost vector $\CostVector$. The cost gaps are the differences between adjacent entries of the sorted cost vector, and our analysis hinges on how these cost gaps are repeated. For our main case where all types have the same costs, cost gaps are repeated $T$ times over the types and this repetition drives our improved upper bound. 
Other assumptions that reduce the number of unique cost gaps also imply an improved approximation ratio, e.g. an arithmetic sequence of cost values implies a $O(\log(nT))$ approximation ratio (see~\Cref{cor:main-upper} for more details).

We next explore the more-fine grained gaps in power between different options available to the principal. For example, how much better is the best single contract than the best linear contract (in the presence of types)? What fraction of welfare is the best menu guaranteed to achieve? Specifically, we consider the following five benchmarks for the principal, in order from least powerful (linear contracts) to most powerful (welfare). 


\begin{itemize}[noitemsep]
  \item \textbf{Linear Contract.} The expected profit of the best linear contract. We write this benchmark as $\OLinear$.
    \item \textbf{Single Contract.} The expected profit of the best single contract. We write this benchmark as $\OSingle$.
  \item \textbf{Menu of Contracts.} The expected profit of the best menu of contracts. We write this benchmark as $\OMenu$.
  \item \textbf{Type-Aware Contract.} The discussion so far assumes that the principal does not (or is not allowed to) observe the type of the agent when designing the mechanism (i.e. the adverse selection asymmetry). If the principal is allowed to offer type-aware contracts, then the best mechanism for the principal is to simply offer each type of agent their optimal contract. We write this benchmark as $\OTypeAware$.
  \item \textbf{Welfare.} The maximum possible profit the principal can hope to obtain is upper bounded by the maximum total expected utility possible -- the expected welfare. This is also the profit achievable when the principal can observe both the type of the agent and the action they take (the ``first-best'' benchmark). We write this benchmark as $\OWelfare$.
\end{itemize}

We characterize the optimal worst-case ratio between nearly all pairs of benchmarks: the results are summarized in Table \ref{tab:ratios}. For all but two of the pairs of benchmarks, we show that the optimal worst-case ratio is equal to $\Theta(n\log T)$; that is, this $\Theta(n\log T)$ separation holds not only between $\OLinear$ and $\OWelfare$, but at a finer-grained level between most pairs of intermediate benchmarks. One exception to this is the gap between $\OTypeAware$ and $\OWelfare$, which is $\Theta(n)$ (this follows directly from the definition of $\OTypeAware$ and the $O(n)$ approximation ratio of \cite{dutting2019simple}). The other (possible) exception is the gap between $\OSingle$ and $\OMenu$; this gap we can only prove is between $2$ (see~\Cref{lem:gap3v4}) and $O(n \log T)$. Understanding the true value of this gap is an interesting open problem, and the last remaining piece of Table \ref{tab:ratios}:

\begin{problem}[Power of menus]
What is the worst-case gap between the expected profit of the best single contract ($\OSingle$) and the expected profit of the best menu of contracts ($\OMenu$)? Does the best single contract achieve a constant-factor approximation (independent of $n$, $m$, and $T$) to the best menu of contracts?
\end{problem}

One may also observe that the number of outcomes $m$ does not appear in Table \ref{tab:ratios}. This is not an oversight: all of our upper bounds (specifically Theorem \ref{thm:linear_main}) works for any number of outcomes $m \geq 1$, and all of our lower bound counterexamples only require a constant number of outcomes ($m \leq 5$). In other words, the gaps between different benchmarks do not seem to scale meaningfully in the number of outcomes. 


\subsection{Related Work}


Contract Theory is a fundamental branch of economics and has been applied in a wide variety of contexts (see~\cite{nobel2016oliver}). For a general introduction to contract theory, we refer the interested reader to the books of~\cite{bolton2005contract,laffont2009theory} and the references therein. While the phenomena of moral hazard and adverse selection have both been studied extensively in the contract theory literature, to the best of our knowledge the interaction between these two phenomena has been significantly less explored. One such example is the work \cite{Laffont1986linked} which consider the principal-agent model with types where the types and actions are linked. Another (more closely related) example is the work of \cite{gottlieb2015simple}, where the case where the types and effort levels can be arbitrarily related.  Their work characterizes a number of different assumptions under which the optimal menu consists of a single contract. 

Our work is most closely related (and in some sense is a direct successor) to the work of \dutting et al. (\cite{dutting2019simple, dutting2020complexity}), which is the first paper we are aware of to study contract theory with the tools of approximation algorithms and computational complexity. In \cite{dutting2019simple}, the authors show that linear contracts are ``robust'' in a certain sense and that they can achieve an $O(n)$-approximation to welfare. In \cite{dutting2020complexity}, the authors show that there exist natural principal-agent problems with exponentially-sized (but succinctly representable) outcome spaces for which finding the optimal contract is computationally hard. 

Some earlier research studies specific variants of contract problems from a computational perspective. One such line of research is the work of \cite{babaioff2006agency} and the related work of \cite{babaioff2006mixed,babaioff2009free}. They introduce a model known as ``Combinatorial Agency'' where a team of agents needs to be incentivized to perform a set of hidden actions, where a combination of their actions yields a certain outcome. Another such line of work studies the problem of learning optimal contracts in a repeated setting \cite{ho2016adaptive}. In \cite{xiao2019optimal}, the authors prove computational results for agents with multiple types (analogous to the results we show in Section \ref{sec:computation}) but for the restricted setting where the principal can contract directly on the agents' actions (we additionally show an APX-hardness result over the NP-hardness result in their paper). 

\newcommand{\multirowrotate}[2]{\multirow{#1}{*}{\rotatebox[origin=c]{90}{#2}}}

\begin{table}
\begin{tabular}{r|ccccc}
  & \multirowrotate{7}{$\OWelfare$}
  & \multirowrotate{7}{$\OTypeAware$}
  & \multirowrotate{7}{$\OMenu$}
  & \multirowrotate{7}{$\OSingle$}
  & \multirowrotate{7}{$\OLinear$} \\
  \\
  \\
  \\
  \\
  \\
  \\
  \midrule
  $\OWelfare$         & $=$                &                    &                           &                    &     \\
  $\OTypeAware$ & $\Theta(n)$        & $=$                &                           &                    &     \\
  $\OMenu$            & $\Theta(n \log T)$ & $\Theta(n \log T)$ & $=$                       &                    &     \\
  $\OSingle$          & $\Theta(n \log T)$ & $\Theta(n \log T)$ & $O(n \log T)$ and $\ge 2$ & $=$                &     \\
  $\OLinear$          & $\Theta(n \log T)$ & $\Theta(n \log T)$ & $\Theta(n \log T)$        & $\Theta(n \log T)$ & $=$ \\
  \bottomrule
\end{tabular}
\caption{Table of our results concerning worst-case ratios between our five benchmarks.}
\label{tab:ratios}
\end{table}

\subsection{Organization}
In~\Cref{sec:prelims}, we set up the problem formally and introduce the five benchmarks that we will use to
evaluate the principal's strategy. In~\Cref{sec:linear}, we will show our main result that 
linear contracts are $\Theta(n \log T)$-competitive with respect to the optimal welfare. We next show how 
to use this reduction combined with a set of reductions to get tight results for almost all the benchmark comparisons in~\Cref{sec:reductions}. We complement these results by showing APX-hardness and exponential time algorithms to compute the optimal menu/contract in~\Cref{sec:computation}.

 \begin{table}
\centering
\begin{tabular}{ccc}
  \toprule
  Measure                   & Running Time                             & Hardness \\ \midrule
  \OWelfare                 & $O(nmT)$                                 & - \\
  \OTypeAware               & $\poly(n, m, T)$                         & - \\
  \multirow{2}{*}{\OMenu}   & $\poly(n^T,m)$                        & \ApxHard with $n = 5$ \\
                            &                                          & (\Cref{thm:apx-hard}) \\
  \multirow{2}{*}{\OSingle} & $\min(\poly(n^T,m), \poly(n^2T)^m))$ & \ApxHard with $n = 5$ \\
                            & (\Cref{lem:algorithm})                   & (\Cref{thm:apx-hard}) \\
  \OLinear                  & $\poly(n,m,T)$                           & - \\
  \bottomrule
\end{tabular}
\caption{Table of our results concerning computability of our five measures.}
\label{tab:computability}
\end{table}

\section{Preliminaries}\label{sec:prelims}

Throughout this paper, we use $[x]$ to denote the one-indexed set $\{1, 2, \dots, x\}$.

\subsection{The Classical Principal-Agent Problem}

In the classical (discrete) principal-agent problem, there is an agent with $n$ actions and $m$ possible (disjoint) outcomes of these actions. Each action $i \in [n]$ incurs a cost $\Cost{i} \geq 0$ for the agent, and each outcome $j \in [m]$ endows the principal with some reward $\Reward{j} \geq 0$. Taking action $i$ leads to a distribution $\D_{i}$ over outcomes; we write $\Forecast{i}{j}$ for the (forecasted) probability that outcome $j \in [m]$ arises as a result of the agent taking action $i \in [n]$. The principal can observe the eventual outcome $j \in [m]$, but not the action $i \in [n]$ that the agent takes (in the contract theory literature, this is referred to as the \emph{hidden-action model}). Without loss of generality, we will always assume the actions are ordered by increasing cost ($\Cost{1} \le \Cost{2} \le \cdots \le \cost{n}$) and outcomes are ordered by increasing reward ($\Reward{1} \le \Reward{2} \le \cdots \le \Reward{m}$). We further assume that there is a zero-cost action ($c_1 = 0$) and a zero-reward outcome ($r_1 = 0$); see \Cref{subsec:simplifying} for why we can do so without loss of generality. We typically group these rewards together into a reward vector $\RewardVector$, these forecasts together into a forecast matrix $\ForecastMatrix$, and these costs together into a cost vector $\CostVector$. These three things constitute the principal-agent problem $\PrincipalAgentProblem$.

The goal of the principal in this problem is to incentivize the agent to perform actions which lead to beneficial outcomes for the principal. To this end, the principal can offer the agent a \emph{contract}. A contract is specified by an $m$-dimensional vector $\ContractVector$, whose nonnegative entries $\Contract{j} \geq 0$ denotes the amount the principal will transfer to the agent in the event of outcome $j$. Given a contract $\ContractVector$, the agent chooses an action which maximizes their overall utility; i.e. the action $\OptimalAction(\ContractVector)$ given by
\begin{align*}
  \OptimalAction(\ContractVector) \in \argmax \left( \sum_{j=1}^m \Forecast{i}{j} \Contract{j} \right) - \Cost{i} \text{.}
\end{align*}

The principal's expected profit is their expected reward minus the expected transfer:
\begin{align*}
  \Profit(\CostVector, \ForecastMatrix, \RewardVector, \ContractVector) \triangleq \sum_{j=1}^m \Forecast{\OptimalAction(\ContractVector)}{j} (\Reward{j} - \Contract{j}).
\end{align*}
Note that in the instance where multiple actions maximize overall utility, we will tiebreak in favor of the principal and assume that the one that results in the highest profit is chosen. The principal can always approximate this tie-breaking behavior with arbitrary accuracy by adding a vanishing fraction of their rewards to the contract. Sometimes, when the principal-agent problem $\PrincipalAgentProblem$ is known from context, we will abbreviate $\Profit(\CostVector, \ForecastMatrix, \RewardVector, \ContractVector)$ to $\Profit(\ContractVector)$. The principal would like to choose a contract $\ContractVector$ to maximize $\Profit(\ContractVector)$. It is known that this can be solved in polynomial time  by solving $O(n)$ linear programs (each computing the minimum expected transfers to force action $i \in [n]$), but this may lead to strange or unnatural contracts that do not typically appear in practice.

One way to try to avoid the issue of strange or unnatural contracts is to restrict the principal to \emph{linear contracts}, where the principal simply transfers some fraction $\TransferCoefficient$ of their reward to the agent. In other words, linear contracts can be completely specified by a transfer coefficient $\TransferCoefficient \in [0, 1]$: $\Contract{j} = \TransferCoefficient \Reward{j}$. When considering linear contracts, we will often abuse notation and use $\TransferCoefficient$ to refer to the contract itself, writing $\OptimalAction(\TransferCoefficient)$ and $\Profit(\TransferCoefficient)$ to denote the optimal action and profit under the contract $\ContractVector = \TransferCoefficient \RewardVector$ respectively.

Another important quantity in the classical principal-agent problem is the maximum achievable \emph{welfare} (sum of principal profit and agent utility). Since by assumption the agent always has a zero-cost action to fall back on, the (maximum) expected welfare is an upper bound on the principal's expected profit. If we let $\ExpectedReward{i}$ denote the expected reward of the principal from the agent playing action $i$, i.e.
\begin{align*}
  \ExpectedReward{i} \triangleq \sum_{j=1}^m \Forecast{i}{j} \Reward{j}\text{,}
\end{align*}
then we can write the welfare $\OWelfare$ as
\begin{align*}
  \OWelfare \PrincipalAgentProblem
    \triangleq \max_i \left( \sum_{j=1}^m \Forecast{i}{j} \Reward{j} - \Cost{i} \right)
    = \max_i \left( \ExpectedReward{i} - \Cost{i} \right)
\end{align*}

\subsection{Typed Agents}

In this paper, we want to incorporate an adverse selection information asymmetry: the principal is uncertain about the exact properties of the agent. Specifically, they are uncertain about the probabilities $\Forecast{i}{j}$ of action $i$ resulting in outcome $j$. We model this by supposing there are $T$ different types of agent.  Although all our results generalize to the setting where the agent types are drawn from an arbitrary discrete distribution, we assume each agent occurs with the same probability to simplify our exposition. Each type $t \in [T]$ has its own set of distributions $\left\{\D_i^{(t)}\right\}_i$ over outcomes with associated probabilities $\TypedForecast{t}{i}{j}$, the probability that a type $t$ agent playing action $i$ results in outcome $j$. Unless otherwise specified, we assume that the cost $\Cost{i}$ of action $i \in [n]$ remains constant over types\footnote{We can model type-dependent costs with type-independent costs at the cost of additional actions. To do so, view the costs of a particular type as a multiset, and let the type-independent costs be the minimal super-multiset of these multisets. For each type's newly introduced costs, let their corresponding dummy actions have the the same probability distribution as that type's zero-cost action, ensuring that they are dominated by an existing action. For a detailed discussion of the effect of this reduction to the approximation ratios, see~\Cref{apx:nonuniform-costs}.}. This assumption is equivalent to saying that the type of an agent has no effect on what effort levels are available, just how effective each effort level is. As before, the principal only observes the eventual outcome $j$; not the action $i$ taken by the agent (hidden-action), nor the type of the agent $t \in [T]$ (adverse selection).

Again, the principal could offer the agent a single contract parameterized by a transfer vector $\ContractVector$. An agent of type $t$ will play the action $\TypedOptimalAction{t}(\ContractVector)$ given by
\begin{align*}
  \TypedOptimalAction{t}(\ContractVector) \in \argmax_i \left( \sum_{j=1}^m \TypedForecast{t}{i}{j} \Contract{j} \right) - \Cost{i} \text{.}
\end{align*}

We assume that each type $t$ occurs in the population with some probability. In this case the principal's profit from an agent of type $t$ is given by
\begin{align*}
  \TypedProfit{t}(\CostVector, \ForecastTensor, \RewardVector, \ContractVector)
    \triangleq \sum_{j=1}^m \TypedForecast{t}{\TypedOptimalAction{t}(\ContractVector)}{j} (\reward{j} - \contract{j})
\end{align*}
and their overall expected profit is given by 
\begin{align*}
  \Profit(\CostVector, \ForecastTensor, \RewardVector, \ContractVector)
    \triangleq \E_t \left[ \TypedProfit{t}(\CostVector, \ForecastTensor, \RewardVector, \ContractVector) \right] \text{.}
\end{align*}
Again, we will tiebreak between multiple utility-maximizing actions in favor of maximizing profit. Unless otherwise specified, for simplicity  we assume that all types occur with equal probability, i.e. the distribution over types is uniform. Again, if $\ContractVector$ is a linear contract with $\TransferCoefficient$ transfer coefficient, we abuse notation and write $\TypedOptimalAction{t}(\TransferCoefficient)$ and $\Profit(\TransferCoefficient)$ for the utility-maximizing action taken by a type $t$ agent and the total profit respectively.

With the addition of these types, a general principal can sometimes secure additional profit if instead they offer the agent a menu of contracts $\ContractMatrix = \left(\TypedContractVector{1}, \ldots, \TypedContractVector{t}\right)$, where an agent who reports their type to be $t$ receives the contract $\TypedContractVector{t}$. We say that such a menu is ``incentive-compatible'' (IC) if no type $t$ has an incentive to misreport their type as a different type $t' \ne t$:
\begin{align}
  \label{ineq:ic}
  \forall t, t' \in [T] \quad \max_i \left( \sum_{j=1}^m \TypedForecast{t}{i}{j} \TypedContract{t}{j} \right) - c_i
                          \ge \max_i \left( \sum_{j=1}^m \TypedForecast{t}{i}{j} \TypedContract{t'}{j} \right) - c_i
  \text{.}
  \tag{IC}
\end{align}

The principal's overall expected profit is given by
\begin{align*}
  \Profit(\CostVector, \ForecastTensor, \RewardVector, \ContractMatrix)
    \triangleq \E_t \left[ \TypedProfit{t}(\CostVector, \ForecastTensor, \RewardVector, \TypedContractVector{t}) \right] \text{.}
\end{align*}

We again define expected reward and welfare, the latter of which still upper-bounds the principal's expected profit:
\begin{align*}
  \TypedExpectedReward{t}{i} &\triangleq \sum_{j=1}^m \TypedForecast{t}{i}{j} \Reward{j}\text{,} \\
  \OWelfare \PrincipalAgentProblem
    &\triangleq \E_t \left[ \max_i \sum_{j=1}^m \TypedForecast{t}{i}{j} \Reward{j} - \Cost{i} \right]
    = \E_t \left[ \max_i \TypedExpectedReward{t}{i} - \Cost{i} \right]
    \text{.}
\end{align*}

\subsection{Five Benchmarks for Typed Principal-Agent Problems}

We summarize the definitions of the five benchmarks we consider for principal-agent problems with types:
\begin{enumerate}
  \item \OWelfare{}, the maximum expected welfare. This is the maximum expected profit attainable by a principal that can observe (and contract on) both types and actions.
  \[
    \OWelfare \PrincipalAgentProblem
      \triangleq \E_t \left[ \max_i \TypedExpectedReward{t}{i} - \Cost{i} \right]
  \]
  \item \OTypeAware{}, the expected profit of the best contract per type. This is the maximum profit attainable by a principal that can observe the types of agents (i.e. if the contracts he assigns can depend on the type of the agent).
  \[
    \OTypeAware \PrincipalAgentProblem
      \triangleq \max_{(\TypedContractVector{1}, \TypedContractVector{2}, \ldots, \TypedContractVector{t})}
        \E_t \left[ \TypedProfit{t}(\CostVector, \ForecastTensor, \RewardVector, \TypedContractVector{t}) \right]
  \]
  \item \OMenu{}, the expected profit of the best (incentive-compatible) menu of contracts.
  \[
    \OMenu \PrincipalAgentProblem
      \triangleq \max_{\ContractMatrix = (\TypedContractVector{1}, \ldots, \TypedContractVector{t}) \text{ is IC}}
        \E_t \left[ \TypedProfit{t}(\CostVector, \ForecastTensor, \RewardVector, \TypedContractVector{t}) \right]
  \]
  \item \OSingle{}, the expected profit of the best single contract.
  \[
    \OSingle \PrincipalAgentProblem
      \triangleq \max_{\ContractVector}
        \E_t \left[ \TypedProfit{t}(\CostVector, \ForecastTensor, \RewardVector, \ContractVector) \right]
  \]
  \item \OLinear{}, the expected profit of the best single linear contract.
  \[
    \OLinear \PrincipalAgentProblem
      \triangleq \max_{\TransferCoefficient \in [0, 1]}
        \E_t \left[ \TypedProfit{t}(\CostVector, \ForecastTensor, \RewardVector, \ContractVector = \TransferCoefficient \RewardVector) \right]
  \]
\end{enumerate}

Each benchmark is an upper bound on the next:
\begin{align*}
  \OWelfare\PrincipalAgentProblem &\ge \OTypeAware\PrincipalAgentProblem \\
                     &\ge \OMenu\PrincipalAgentProblem\\
                     &\ge \OSingle\PrincipalAgentProblem \\
                     &\ge \OLinear\PrincipalAgentProblem \text{.}
\end{align*}

\section{Linear Contracts}
\label{sec:linear}

In this section, we explore the surprising power of linear contracts, in which the principal gives a constant $\TransferCoefficient$ fraction of their reward to the agent. Our main goal is to show that the profit of the optimal linear contract is always within a $O(n \log T)$ factor of the optimal welfare and that this is tight; there exists a principal-agent problem $\PrincipalAgentProblem$ where the best linear contract achieves profit that is a $\Omega(n \log T)$ factor from the optimal welfare.

Central to both our upper and lower bounds are the following two observations about linear contracts. Recall our definition of the expected reward when a type $t$ agent plays action $i$: $\TypedExpectedReward{t}{i} \triangleq \sum_{j=1}^m \TypedForecast{t}{i}{j} \Reward{j}$. The first observation is that $\OWelfare$ and $\Profit(\TransferCoefficient)$ can be written entirely in terms of $\TransferCoefficient$, $\CostVector$, and the $\TypedExpectedReward{t}{i}$. In particular, the values of $\TypedForecast{t}{i}{j}$ and $\Reward{j}$ \emph{do not affect either of these measures} beyond their use in the computation of $\TypedExpectedReward{t}{i}$. Formally, we have that (recall that $\TypedOptimalAction{t}(\TransferCoefficient)$ is the utility-maximizing action taken by type $t$ under the linear contract $\alpha$):
\begin{align*}
  \OWelfare                     &= \E_t \left[ \max_i \left(\TypedExpectedReward{t}{i} - \Cost{i} \right) \right] \text{ and} \\
  \Profit(\TransferCoefficient) &= \E_t \left[ (1-\TransferCoefficient) \TypedExpectedReward{t}{\TypedOptimalAction{t}(\TransferCoefficient)} \right] \text{.}
\end{align*}

The second observation is that $\Profit(\TransferCoefficient)$ is a \emph{piecewise linear function with at most $nT$ pieces}. This stems from the fact that each $\TypedProfit{t}(\TransferCoefficient)$ is a piecewise-linear function with at most $n$ pieces, which is because for each range of $\alpha$ where $\TypedOptimalAction{t}(\TransferCoefficient)$ is constant, $\TypedProfit{t}(\TransferCoefficient)$ is linear (and there are at most $n$ possible values of $\TypedOptimalAction{t}(\TransferCoefficient)$). 

In fact, we can say something even stronger. Let's consider the resulting utility of a type $t$ agent when offered a linear contract with transfer coefficient $\TransferCoefficient$ and the expectation of this quantity over types:
\begin{align*}
  \Utility^{(t)}(\TransferCoefficient) &\triangleq \max_i \left( \alpha \TypedExpectedReward{t}{i} - \Cost{i} \right)\text{ and} \\
  \Utility(\TransferCoefficient)       &\triangleq \E_t \left[ \Utility^{(t)}(\TransferCoefficient) \right] \text{.}
\end{align*}
By similar logic, $\Utility(\TransferCoefficient)$ is a piecewise linear function with at most $nT$ pieces. Moreover, from the definitions of profit and utility we have that:
\begin{align}
  \label{eq:profit-utility}
  \Profit(\TransferCoefficient) = (1-\TransferCoefficient) \Derivative{\TransferCoefficient} \Utility(\TransferCoefficient)\text{.}
\end{align}

Since $\OWelfare = \Utility(1)$, \Cref{eq:profit-utility} allows us to relate the expected profit of the optimal linear contract to the maximum expected welfare and phrase our original problem as an optimization problem over the piecewise linear functions $\Utility$. \Naively{} solving this optimization problem (over all increasing piecewise linear functions $\Utility$ with $nT$ pieces) leads to an $O(nT)$-approximation factor, but this does not use any information about types\footnote{It may be instructive to compare this with the \dutting{} et al. proof of the $O(n)$ approximation factor of linear contracts, when types are not an issue \cite{dutting2019simple}.}. Since each type has the same set of allowable costs, this imposes additional constraints on $\Utility$. Specifically, the linear segments of $\Utility$, if extended, can have at most $T$ different ``$y$-intercepts''. We leverage these constraints to give an improved $O(n \log T)$-approximation ratio.

\subsection{Linear Contracts: Profit Guarantee}

In this subsection, we prove that the optimal linear contract obtains profit within a $O(n \log T)$ factor of $\OWelfare$. In the proof of the following theorem, it will be helpful to assume that all actions are played by all types for some choice of $\TransferCoefficient \in [0, 1]$, i.e. for every type $t \in [T]$ and action $i \in [n]$, there exists some $\TransferCoefficient \in [0, 1]$ such that $\TypedOptimalAction{t}(\TransferCoefficient) = i$. This holds without loss of generality because if action $i$ is never chosen by type $t$, we can increase $\TypedExpectedReward{t}{i}$ until it is barely chosen; this affects neither $\OWelfare$ nor $\OLinear$. 

\begin{theorem}
\label{thm:main-upper}
For any principal-agent problem $\PrincipalAgentProblem$,
\begin{align*}
  \OLinear \PrincipalAgentProblem \ge \Omega \left(\frac{1}{n \log T}\right) \OWelfare \PrincipalAgentProblem\text{.}
\end{align*}
\end{theorem}

\newcommand{\Breakpoint}[1]{\alpha_{#1}}
\newcommand{\TypedBreakpoint}[2]{\Breakpoint{#2}^{(#1)}}
\newcommand{\CoefficientJump}[1]{w_{#1}}
\newcommand{\TypedCoefficientJump}[2]{\CoefficientJump{#2}^{(#1)}}
\newcommand{\CumulativeCoefficientJump}[1]{W_{#1}}
\newcommand{\BiasDrop}[1]{\delta_{#1}}
\newcommand{\TypedBiasDrop}[2]{\BiasDrop{#2}^{(#1)}}
\newcommand{\CumulativeBiasDrop}[1]{D_{#1}}

\newcommand{\MagicDifference}[1]{v_{#1}}
\newcommand{\CumulativeMagicDifference}[1]{V_{#1}}

\newcommand{\MagicRatio}[1]{\beta_{#1}}

\begin{proof}
Rewriting the theorem statement, we must show that there exists an $\TransferCoefficient \in [0, 1]$ such that $\Profit(\alpha) \ge \frac{\OWelfare}{O(n\log T)}$.

We begin by considering our utility and profit functions. Note that $\Utility^{(t)}(\TransferCoefficient)$ is an increasing piecewise linear function with $n$ pieces (one for each possible action). We denote its $i^{th}$ breakpoint as $\TypedBreakpoint{t}{i}$, which is where $\Utility^{(t)}(\TransferCoefficient)$ changes from $\TransferCoefficient \TypedExpectedReward{t}{i} - \Cost{i}$ to $\TransferCoefficient \TypedExpectedReward{t}{i+1} - \Cost{i+1}$. We define some notation to express how much the multiplier on $\TransferCoefficient$ increases and how much the constant term decreases:
\begin{align*}
  \TypedCoefficientJump{t}{i} &\triangleq \TypedExpectedReward{t}{i+1} - \TypedExpectedReward{t}{i} \\
  \TypedBiasDrop{t}{i} &\triangleq \Cost{i+1} - \Cost{i}
\end{align*}
One way of looking at the situation is that when we increase $\TransferCoefficient$ past $\TypedBreakpoint{t}{i}$, we add the linear function $\TypedCoefficientJump{t}{i} \TransferCoefficient - \TypedBiasDrop{t}{i}$ to $\Utility^{(t)}(\TransferCoefficient)$.

The function $\Utility(\TransferCoefficient)$ is the convex combination of $T$ different functions $\Utility^{(t)}(\TransferCoefficient)$. It therefore has $(n-1)T$ breakpoints: the union of breakpoints for all the functions $\Utility^{(t)}(\TransferCoefficient)$. Sort these $B \triangleq (n-1)T$ breakpoints in increasing order. Abusing notation, we now call the $i^{th}$ such breakpoint $\Breakpoint{i}$, its corresponding multiplier increase $\CoefficientJump{i}$, and its corresponding constant term decrease $\BiasDrop{i}$. Note that each such pair of values $\left( \CoefficientJump{i}, \BiasDrop{i} \right)$ previously appeared (a factor of $T$ larger due to expectation over types) as a corresponding pair of values $\left( \TypedCoefficientJump{t}{i'}, \TypedBiasDrop{t}{i'} \right)$ of one of the original functions $\Utility^{(t)}$. In particular, note that there can only be $n$ distinct values of $\BiasDrop{i}$, since $\BiasDrop{t}{i'}$ does not depend on $t$.

Next, consider $\TypedProfit{t}(\TransferCoefficient)$. Recall that by the definition of $\TypedProfit{t}$, $\TypedProfit{t}(\TransferCoefficient) = (1 - \TransferCoefficient) \Derivative{\TransferCoefficient} \Utility^{(t)}(\TransferCoefficient)$ (at a breakpoint we can arbitrarily choose the larger of the two derivatives). Since differentiation is linear, this property also holds for $\Profit$: $\Profit(\TransferCoefficient) = (1 - \TransferCoefficient) \Derivative{\TransferCoefficient} \Utility(\TransferCoefficient)$. Since $\Profit$ is piecewise linear, its maximum will occur when $\TransferCoefficient$ equals some breakpoint $\Breakpoint{i}$, zero, or one. In fact, it cannot occur when $\TransferCoefficient$ is zero but not a breakpoint since there is no incentive for the agent to pick a non-null action, and it cannot occur when $\TransferCoefficient$ is one because the principal gives all reward to the agent. We now relate $\Profit$ to other quantities.

We define the cumulative terms $\CumulativeCoefficientJump{i} \triangleq \sum_{j = 1}^{i} \CoefficientJump{j}$ and $\CumulativeBiasDrop{i} \triangleq \sum_{j=1}^{i} \BiasDrop{j}$. Since we have a dummy action, we know that for $\TransferCoefficient \in [\Breakpoint{i}, \Breakpoint{i+1})$, $\Utility(\TransferCoefficient) = \CumulativeCoefficientJump{i} \TransferCoefficient - \CumulativeBiasDrop{i}$ and therefore $\Profit(\TransferCoefficient) = \CumulativeCoefficientJump{i} (1 - \TransferCoefficient)$. Since the latter is decreasing in $\TransferCoefficient$, in this range, it is maximized at $\TransferCoefficient = \Breakpoint{i}$ at which $\Profit(\Breakpoint{i}) = \CumulativeCoefficientJump{i} (1 - \Breakpoint{i})$.

What is $\Breakpoint{i}$? Well, consider the point $\TransferCoefficient = \Breakpoint{i}$, where some agent is ambivalent between two actions.
\begin{align*}
  \CumulativeCoefficientJump{i} \Breakpoint{i} - \CumulativeBiasDrop{i} &= \CumulativeCoefficientJump{i-1} \Breakpoint{i} - \CumulativeBiasDrop{i-1} \\
  \CoefficientJump{i} \Breakpoint{i} &= \BiasDrop{i} \\
  \Breakpoint{i} &= \frac{\BiasDrop{i}}{\CoefficientJump{i}}
\end{align*}

Plugging this in, we get
\begin{equation}
\label{eq:rev-form1}
  \Profit(\Breakpoint{i}) = \CumulativeCoefficientJump{i} \left(1 - \frac{\BiasDrop{i}}{\CoefficientJump{i}} \right)
  \text{.}
\end{equation}

Remember that we need to compare this quantity to the welfare $\Welf = \Utility(1) = \CumulativeCoefficientJump{B} - \CumulativeBiasDrop{B}$. To make this comparison, we rewrite both $\Welf$ and \Cref{eq:rev-form1} in terms of a cumulative difference $\CumulativeMagicDifference{i}$ instead of $\CumulativeCoefficientJump{i}$. Define the following differences:
\begin{align*}
  \MagicDifference{i}           &\triangleq \CoefficientJump{i} - \BiasDrop{i} \\
  \CumulativeMagicDifference{i} &\triangleq \sum_{j=1}^i \MagicDifference{j} = \CumulativeCoefficientJump{i} - \CumulativeBiasDrop{i}
\end{align*}

We can now write welfare as $\Welf = \CumulativeMagicDifference{B}$ and our profit as:
\begin{align*}
  \Profit(\Breakpoint{i})
    &= (\CumulativeMagicDifference{i} + \CumulativeBiasDrop{i}) \left( 1 - \frac{\BiasDrop{i}}{\MagicDifference{i} + \BiasDrop{i}} \right) \\
    &= (\CumulativeMagicDifference{i} + \CumulativeBiasDrop{i}) \frac{\MagicDifference{i}}{\MagicDifference{i} + \BiasDrop{i}} \\
    &= \frac{\CumulativeMagicDifference{i} + \CumulativeBiasDrop{i}}{1 + \BiasDrop{i} / \MagicDifference{i}} \\
    &\ge \min \left\{ \CumulativeMagicDifference{i}, \frac{\CumulativeBiasDrop{i}}{\BiasDrop{i}} \MagicDifference{i} \right\} \text{.}
\end{align*}

It now suffices to show that there exists an $i$ such that $\min \left\{ \CumulativeMagicDifference{i}, \frac{\CumulativeBiasDrop{i}}{\BiasDrop{i}} \MagicDifference{i} \right\} \geq \CumulativeMagicDifference{B} / O(n \log T)$. We do this as follows. First filter out all the $i$ such that $\CumulativeMagicDifference{i} < \CumulativeMagicDifference{B} / O(n \log T)$. The sum of the remaining $\MagicDifference{i}$ is at least $(1 - 1/O(n \log T)) \CumulativeMagicDifference{B} \ge \tfrac12 \CumulativeMagicDifference{B}$. Now we simply need to show that there exists one $i$ such that
\begin{align*}
  \frac{\CumulativeBiasDrop{i}}{\BiasDrop{i}} \MagicDifference{i} \geq \frac{1}{O(n\log T)} \CumulativeMagicDifference{B}\text{.}
\end{align*}

Let's consider these multiplier terms $\MagicRatio{i} = \frac{\CumulativeBiasDrop{i}}{\BiasDrop{i}}$. We would like to argue that they are quite large. To do so, we use the fact that there are at most $n$ possible values of $\BiasDrop{i}$. Specifically, let's group the remaining $i$'s together based on their original indices $i'$ (recall that each $\BiasDrop{i}$ is just $\tfrac1T \TypedBiasDrop{t}{i'}$). This creates at most $n$ groups of $i$'s, each with at most $T$ elements. We can group $i$'s together in an alternate way; let $S_k$ be the set of indices $i$ such that indices that derive from the original index $i'$ and have been seen \emph{exactly} $k$ times up to and including index $i$ (i.e. $|\{j \mbox{ s.t. } j \leq i \mbox{ and both } i, j \text{ have original index } i'\}| = k$). We know the following:
\begin{enumerate}
  \item For all $k > T$, $S_k = \emptyset$ (each original index $i'$ generates $T$ indices $i$).
  \item For all $k$, $|S_k| \leq n$ (there are only $n$ original indices $i'$).
  \item For $i \in S_k$, $\MagicRatio{i} \ge k$ (since $\CumulativeBiasDrop{i} \ge k \BiasDrop{i}$ at this point).
\end{enumerate}

Let $\sigma_k = \sum_{i \in S_k} \MagicDifference{i}$. By (1), we know that $\sum_{k=1}^T \sigma_k \ge \tfrac12 \CumulativeMagicDifference{B}$. It follows that there exists an $k$ such that $k \sigma_k \geq \frac{\CumulativeMagicDifference{B}}{O(\log T)}$; if there weren't, then we would have $\sigma_k \le o\left(\frac{\CumulativeMagicDifference{B}}{k\log T}\right)$ for all $k$ and summing over $k$ we would have that $\sum_k \sigma_k \leq o(1) \CumulativeMagicDifference{B}$, a contradiction. By (2), it follows that there exists a $i \in S_k$ such that $k \MagicDifference{i} \ge \frac{\CumulativeMagicDifference{B}}{O(n\log T)}$. Finally, by (3) it follows that $\MagicRatio{i} \MagicDifference{i} \geq \frac{\CumulativeMagicDifference{B}}{O(n\log T)}$, as desired. This completes the proof.
\end{proof}

A closer examination of the above proof shows that the approximation ratio we obtain only depends on the multiset of the values $\delta_i$ (i.e., the differences between $c_{i+1} - c_{i}$ for consecutive costs). Let $S$ be this multi-set. Assume that $S$ has $U$ distinct values which appear with multiplicities $m_1, m_2, \dots m_U$. Let $\Delta(\CostVector{}) = \sum_{i=1}^{U} \log m_U$. Then a direct generalization of the proof of Theorem \ref{thm:main-lower} gives the following, more general corollary.

\begin{corollary}
\label{cor:main-upper}
For any (possibly cost-varying) principal-agent problem $\PrincipalAgentProblem$, let $S(\CostVector{})$ be the multiset of cost differences $\{c^{(t)}_{i+1} - c_{i}^{(t)} \mid i \in [n-1], t \in [T]\}$. Assume $S(\CostVector{})$ has $U$ distinct values which appear with multiplicities $m_1, m_2, \dots, m_U$. Then
\begin{align*}
  \OLinear \PrincipalAgentProblem \ge \Omega \left(\frac{1}{ \Delta(\CostVector{}) }\right) \OWelfare \PrincipalAgentProblem\text{,}
\end{align*}
 where $\Delta(\CostVector{}) = \sum_{i=1}^{U} \log m_i$.
\end{corollary}



This has a number of immediate consequences:

\begin{itemize}
    \item 
    For a standard principal-agent problem with no constraints on the costs $\CostVector{}$, $\Delta(\CostVector) \leq n$. This immediately implies the $O(n)$-approximation ratio proved in \cite{dutting2019simple}.
    \item
    For a standard principal-agent problem where the costs $\CostVector{}$, lie in an arithmetic sequence, $\Delta(\CostVector) = O(\log n)$. In particular, if the set of costs is generated (for example) by discretizing a continuum of actions, we can hope for exponentially better approximation ratios than in the general case.
    \item
    For a typed principal-agent problem with $n$ actions and $T$ types, $\Delta(\CostVector) \leq O(n \log T)$. This recovers the result of Theorem \ref{thm:main-upper} proved above.
    \item 
    For a typed principal-agent problem with $n$ actions and $T$ types whose costs $\CostVector{}$ lie in an arithmetic sequence, $\Delta(\CostVector) = O(\log(nT))$.
    \item
    For a typed principal-agent problem with $n$ actions and $T$ types whose costs $\CostVector{}$ may vary between types, $\Delta(\CostVector) = O(nT)$ (see also Appendix \ref{apx:nonuniform-costs} for a discussion of cost-varying typed principal-agent problems).
\end{itemize}

\subsection{Linear Contracts: Tight Counterexample}

We will now prove that the bound in \Cref{thm:main-upper} is tight. The following example follows almost immediately from the proof of Theorem \ref{thm:main-upper}; we choose a $\Utility(\TransferCoefficient)$ to satisfy all necessary inequalities in the proof of Theorem \ref{thm:main-upper}, and then implement this $\Utility(\TransferCoefficient)$ via a principal-agent problem. 

\begin{theorem}
\label{thm:main-lower}
For all $n, T > 0$, there exists a principal-agent problem $\PrincipalAgentProblem$ with $T$ types, $n+1$ actions, and two outcomes that satisfies
\begin{align*}
  \OLinear \PrincipalAgentProblem \le O\left( \frac{1}{n \log T} \right) \OWelfare \PrincipalAgentProblem
  \text{.}
\end{align*}
\end{theorem}

\begin{proof}
Consider the following principal-agent problem. As stated, we will have $T$ types, $n+1$ actions, and two outcomes. For this proof, we zero-index our actions (zero and one) and outcomes (zero through $n$) Action zero will have cost $\Cost{0} \triangleq 0$, and action $i \in [n]$ will have cost $\Cost{i} \triangleq \sum_{j=1}^i \lambda^j$, where $\lambda$ is a constant to be determined later (we will take $\lambda \rightarrow \infty$).

As for outcomes, outcome zero contributes reward $\Reward{0} = 0$. Outcome one will contribute reward $\Reward{1} = \Cost{n} + \frac{n}{T \log T}$. For each type $t \in [T]$, and action $i \in \{0, \dots, n\}$, define
$$\TypedExpectedReward{t}{i} \triangleq \Cost{i} + \frac{i}{t n \log T}.$$

Let $p_{i}^{(t)} = \frac{\TypedExpectedReward{t}{i}}{\Reward{1}}$. If type $t$ plays action $i$, then outcome one will occur with probability $p_{i}^{(t)}$ and outcome zero will occur with probability $1-p_{i}^{(t)}$. Note that the expected reward for the principal of type $t$ playing action $i$ is exactly $\TypedExpectedReward{t}{i}$.

We begin by computing $\OWelfare$ for this instance. Note that the contribution to welfare from type $t$ is given by
$$\max_{i} (\TypedExpectedReward{t}{i} - \Cost{i}) = \max_{i} \frac{i}{t n \log T} = \frac{1}{t \log T}.$$

Averaging this over all types $t \in [T]$, we have that
$$ \OWelfare = \frac1T \sum_{t=1}^T \frac{1}{t\log T} = \Theta\left(\frac1T\right).$$

It thus suffices to show that $\Profit(\TransferCoefficient) = O(1 / n T \log T)$ for all $\TransferCoefficient \in [0, 1]$. We begin by computing the breakpoints $\TypedBreakpoint{t}{i}$ where the function $\Profit(\TransferCoefficient)$ changes slope. These points satisfy $\alpha_{i}^{(t)}R_{i}^{(t)} - c_{i} = \alpha_{i}^{(t)}R_{i+1}^{(t)} - c_{i+1}$, so
$$\alpha_{i}^{(t)} = \frac{c_{i+1} - c_{i}}{R_{i+1}^{(t)} - R_{i}^{(t)}} = \frac{\lambda^{i+1}}{\lambda^{i+1} + (nt\log T)^{-1}}.$$

For sufficiently large $\lambda$, these $\alpha_{i}^{(t)}$ are ordered first by type and then by action, i.e. $\alpha_{0}^{(1)} < \alpha_{0}^{(2)} < \dots < \alpha_{0}^{(T)} < \alpha_{1}^{(1)} < \dots < \alpha_{1}^{(T)} < \dots < \alpha_{n-1}^{(T)}$. (One way to see this is to notice that $(\alpha_{i}^{(t)})^{-1} = 1 + (nt\lambda^{i+1}\log T)^{-1}$; as long as $\lambda > T$, this will satisfy the previous ordering). Now, for $\alpha \in [\alpha_{i}^{(t)}, \alpha_{i}^{(t+1)}]$, note that the first $t$ types are playing action $i+1$ and the remaining types are playing action $i$. It follows that
\begin{eqnarray*}
  \Profit(\alpha) &=& \frac1T (1-\alpha)\left(\sum_{s=1}^{t} R_{i+1}^{(s)} + \sum_{s=t+1}^{T} R_{i}^{(s)}\right) \\
  &=& \frac1T (1 - \alpha) \left(\sum_{s=1}^{T} R_{i}^{(s)} + \sum_{s=1}^{t} (R_{i+1}^{(s)} - R_{i}^{(s)}) \right) \\
  &=& \frac1T (1 - \alpha) \left(\sum_{s=1}^{T} (c_i + \frac{i}{sn\log T}) + \sum_{s=1}^{t} (\lambda^{i+1} + (sn\log T)^{-1}) \right) \\
  &\leq & \frac1T (1 - \alpha_{i}^{(t)})\left(Tc_i + \frac{i}{n} + t\lambda^{i+1} + \frac{\log t}{n\log T} \right)  \\
  &\leq & \frac1T (1 - \alpha_{i}^{(t)})\left(T\left(\sum_{j=1}^{i}\lambda^{j}\right) + t\lambda^{i+1} + 2 \right)  \\
  &\leq & \frac1T (1 - \alpha_{i}^{(t)})(2t\lambda^{i+1}) \hspace{20mm} \mbox{ for sufficiently large $\lambda$} \\ 
  &=& \frac1T \frac{(2nt\log T)^{-1}}{\lambda^{i+1} + (nt\log T)^{-1}}(t\lambda^{i+1}) \\
  &=& \frac1T \frac{(t\lambda^{i+1})}{nt\lambda^{i+1}\log T + 1} \\
  &\leq & \frac{2}{n T \log T}.
\end{eqnarray*}

\noindent
A similar computation for $\alpha \in [\alpha_i^{(T)}, \alpha_{i+1}^{(1)}]$ shows that in this case as well, $\Profit(\alpha) \leq \frac{2}{n T \log T}$. It follows that $\Profit(\alpha) \leq O(1/(n\log T))\cdot \OWelfare$, as desired.


\end{proof}

\section{Fine Grained Gaps}
\label{sec:reductions}

In this section, we utilize \Cref{thm:main-upper} and \Cref{thm:main-lower} to get our worst-case bounds on the ratios between nearly all pairs of benchmarks.

We first discuss the easy side, namely upper-bounds on these ratios. From \Cref{thm:main-upper}, we know that the expected profit of the optimal linear contract, $\OLinear$, can be at most $O(n \log T)$ away from the maximum expected welfare, $\OWelfare$. Since these are the lowest and highest benchmarks, respectively, we know that the ratio between any pair of benchmarks can be at most $O(n \log T)$. The one pair where we obtain a better bound is between $\OTypeAware$ and $\OWelfare$. In this case, the problem decomposes by type. Since we can always just construct a linear contract for each type, the ratio is at most $O(n)$. As a result, we need no additional machinery to obtain all our ratio upper-bounds.

We now discuss the hard side, namely lower-bounds on these ratios. On this side, we can use lower-bounds on the ratios between closer benchmarks to get lower-bounds on the ratio between further benchmarks. All the lower-bounds in \Cref{tab:ratios} hence follow from lower-bounds on the ratios between adjacent benchmarks.

From \Cref{thm:main-lower}, we know there is a two-outcome family of counterexamples which shows that the ratio between $\OLinear$ and $\OWelfare$ is $\Omega(n \log T)$. In \Cref{subsec:two-outcome}, we show that in two-outcome principal-agent problems, the optimal menu of contracts cannot extract any more profit than the optimal linear contract. In \Cref{subsec:nonlinearity-disparity}, we show that we can edit any counterexample so that single contracts extract the entire welfare as profit but linear contracts do no better. In \Cref{subsec:info-is-power}, we show that we can edit any counterexample so that type-aware principals extract the entire welfare as profit but menus do no better. Combining these ideas yields that the ratio between $\OLinear$ and $\OSingle$ and the ratio between $\OMenu$ and $\OTypeAware$ are both $\Omega(n \log T)$.

This leaves two pairs of adjacent benchmarks. In \Cref{apx:gap3v4}, we directly show that the ratio between $\OSingle$ and $\OMenu$ is at least $(2 - \epsilon)$ for every constant $\epsilon > 0$. As for the ratio between $\OTypeAware$ and $\OWelfare$, the problem again decomposes by type and it was already known that the gap between first-best and second-best outcomes was $\Omega(n)$ \cite{dutting2019simple}.

See \Cref{fig:roadmap} for a helpful depiction of this plan.

\begin{figure}[h]
\centering
\begin{tikzpicture}[%
  auto,
  scale=1.0,
  result/.style={
    rectangle,
    draw=black,
    inner sep=1pt,
  },
  ]
  \node at (-5, 2) {\textbf{Upper Bounds}};
  
  \node[result] (generalizedwelfaregap) at (0, 2) {\begin{tabular}{c}
    \Cref{thm:main-upper} \\
    $\OLinear$ vs. $\OWelfare$ \\
    $O(n \log T)$ gap
  \end{tabular}};
  
  \draw[dashed] (-6, 1) -- (6, 1);
  
  \node at (-5, 0.5) {\textbf{Lower Bounds}};
  
  \node[result] (basecounterexample) at (-3, -1) {\begin{tabular}{c}
    \Cref{thm:main-lower} \\
    $\OLinear$ vs. $\OWelfare$ \\
    $\Omega(n \log T)$ gap with $m = 2$
  \end{tabular}};
  
  \node[result] (observedcounterexample) at (3, -1) {\begin{tabular}{c}
    \Cref{cor:two-outcome-linear} \\
    $\OMenu$ vs. $\OWelfare$ \\
    $\Omega(n \log T)$ gap with $m = 2$
  \end{tabular}};
  
  \node[result] (nonlinearitydisparity) at (-3, -4) {\begin{tabular}{c}
    \Cref{cor:nonlinearity-disparity} \\
    $\OLinear$ vs. $\OSingle$ \\
    $\Omega(n \log T)$ gap with $m = 4$
  \end{tabular}};
  
  \node[result] (infoispower) at (3, -4) {\begin{tabular}{c}
    \Cref{cor:info-is-power} \\
    $\OMenu$ vs. $\OTypeAware$ \\
    $\Omega(n \log T)$ gap with $m = 4$
  \end{tabular}};
  
  \node[result] at (0, -6) {\begin{tabular}{c}
    \Cref{lem:gap3v4} \\
    $\OSingle$ vs. $\OMenu$ \\
    $(2 - \epsilon)$ gap with $n = 3$
  \end{tabular}};
  
  \draw[->] (basecounterexample) -- (observedcounterexample) node[midway, above]
  {\Cref{lem:two-outcome-linear}};
  \draw[->] (basecounterexample) -- (nonlinearitydisparity) node[midway, left] {\Cref{lem:nonlinearity-disparity}};
  \draw[->] (observedcounterexample) -- (infoispower) node[midway, right] {\Cref{thm:info-is-power}};
\end{tikzpicture}
\caption{Proof road map of our upper and lower bound results between our various contract benchmarks.}
\label{fig:roadmap}
\end{figure}
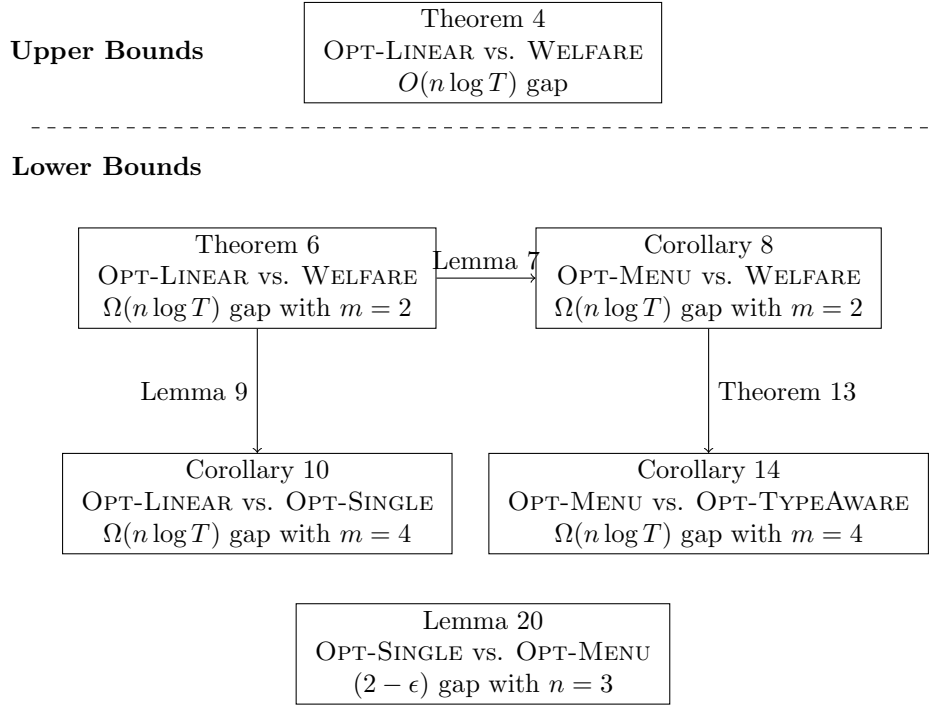

\subsection{The Two Outcome Case}
\label{subsec:two-outcome}

In this subsection, we consider the case where there are only two outcomes ($m = 2$). Intuitively, the only purpose of the contract is to encourage the higher-valued outcome and hence it only makes sense to offer zero on the lower-valued outcome and something on the higher-valued outcome. Given our assumption that the first reward $\reward{1}$ is zero, such a contract is linear. In fact, we might hope that this reasoning applies to every contract in a menu and hence linear contracts are as powerful as menus in this scenario. \Cref{lem:two-outcome-linear} and its proof (found in \Cref{apx:two-outcome}) formalize this intuition.

\begin{restatable}{lemma}{twooutcomelemma}
\label{lem:two-outcome-linear}
  Suppose we have a principal-agent problem $\PrincipalAgentProblem$ with $m = 2$ outcomes. Then linear contracts can extract as much profit as menus of contracts:
  \begin{align*}
    \OMenu \PrincipalAgentProblem = \OLinear \PrincipalAgentProblem\text{.}
  \end{align*}
\end{restatable}

Applying \Cref{lem:two-outcome-linear} to \Cref{thm:main-lower} yields the following corollary.

\begin{corollary}
\label{cor:two-outcome-linear}
For all $n, T > 0$, there exists a principal-agent problem $\PrincipalAgentProblem$ with $T$ types, $n+1$ actions, and two outcomes that satisfies
\begin{align*}
  \OMenu \PrincipalAgentProblem \le O\left( \frac{1}{n \log T} \right) \OWelfare \PrincipalAgentProblem
  \text{.}
\end{align*}
\end{corollary}

\subsection{Reduction: Nonlinearity Disparity}
\label{subsec:nonlinearity-disparity}

In this subsection, we give a transformation for (typed) principal-agent problems that helps the optimal single contract extract the entire welfare as profit without allowing linear contracts to extract more profit.

\begin{lemma}
\label{lem:nonlinearity-disparity}
  Suppose we have a principal-agent problem $\PrincipalAgentProblem$ with $T$ types, $n$ actions, and $m$ outcomes. We can construct another contract problem $\PrincipalAgentProblemPrime$ with $T$ types, $n$ actions, and $m+2$ outcomes such that:
  \begin{align*}
    \OSingle \PrincipalAgentProblemPrime &= \OWelfare \PrincipalAgentProblem \\
    \OLinear \PrincipalAgentProblemPrime &= \OLinear \PrincipalAgentProblem
  \end{align*}
\end{lemma}

\begin{proof}
  Our goal is to permit a single contract to extract the full welfare as profit but not allow linear contracts to extract any additional profit. We will achieve this by adding a new zero-reward outcome that can be used to encourage welfare-maximizing actions; note such linear contracts are blind to such outcomes. We also will need an additional zero-reward action to rebalance probabilities.

  Zooming in, we want to arrange things so that a single contract that only offers a positive transfer on outcome $(m + 1)$ winds up perfectly paying the cost of the welfare-maximizing action. Note that the identity of this action changes with the agent's type $t$; denote the identity of this action for type $t$ with $\TypedOptimalAction{t}$, breaking ties arbitrarily. We will need the probability of outcome $(m + 1)$ to be proportional to the cost of action $\TypedOptimalAction{t}$, since the expected transfers are this probability times the same base transfer amount. Let $\epsilon \triangleq (2 \max_t \Cost{\TypedOptimalAction{t}})^{-1}$, and let for all $t \in [T]$ let $\epsilon_t \triangleq \epsilon \Cost{\TypedOptimalAction{t}}$ which is at most one-half by construction; this outcome will occur for the type $t$'s welfare-maximizing action with probability $\epsilon_t$.
  
  We plan to dedicate one-half probability mass total to our two new outcomes, so outcome $m + 2$ will pick up the remaining probability. We also need to scale down the probability of existing outcomes by a factor of one-half so everything sums to one. To compensate for this probability scaling, we also scale up the rewards by a factor of two.
  
  Combining the above ideas, we formally define our new principal-agent problem $\PrincipalAgentProblemPrime$ as follows (see \Cref{tab:nonlinearity-disparity} for a depiction of this construction in tabular form).
  \begin{alignat*}
    \forall i \in [n] \quad
    &&\Cost{i}' &\triangleq \Cost{i} \\
    \forall t \in [T], i \in [n], j \in [m+2] \quad
    &&\forecastprime{t}{i}{j} &\triangleq \begin{cases}
      \tfrac12 \TypedForecast{t}{i}{j}                                  & \text{ if } j \in [m] \\
      \epsilon_t \indicator{i = \TypedOptimalAction{t}}                 & \text{ if } j = m + 1 \\
      \tfrac12 - \epsilon_t \indicator{i \not = \TypedOptimalAction{t}} & \text{ if } j = m + 2 \\
    \end{cases} \\
    \forall j \in [m+2] \quad
    &&\Reward{j}' &\triangleq \begin{cases}
      2 \Reward{j} & \text{ if } j \in [m] \\
      0            & \text{ if } j \in \{m+1, m+2\}
    \end{cases}
  \end{alignat*}
  
  \begin{table}
\centering
\begin{tabular}{cccccc}
  \toprule
  \multirow{2}{*}{Type $t \in [T]$} & \textbf{Outcome} $1$  & \multirow{2}{*}{$\cdots$} & \textbf{Outcome} $m$  & \textbf{Outcome} $m+1$ & \textbf{Outcome} $m+2$ \\
                                    & Reward $2 \reward{1}$ &                           & Reward $2 \reward{m}$ & Reward $0$             & Reward $0$               \\ \midrule
  \textbf{Action} $1$ & \multirow{2}{*}{$\tfrac12 \TypedForecast{t}{1}{1}$} & \multirow{2}{*}{$\cdots$} & \multirow{2}{*}{$\tfrac12 \TypedForecast{t}{1}{m}$} & \multirow{2}{*}{$0$} & \multirow{2}{*}{$\tfrac12$} \\ Cost $\cost{1}$ \\
  \textbf{Action} $2$ & \multirow{2}{*}{$\tfrac12 \TypedForecast{t}{2}{1}$} & \multirow{2}{*}{$\cdots$} & \multirow{2}{*}{$\tfrac12 \TypedForecast{t}{2}{m}$} & \multirow{2}{*}{$0$} & \multirow{2}{*}{$\tfrac12$} \\ Cost $\cost{2}$ \\
  $\vdots$ & $\vdots$ & $\ddots$ & $\vdots$ & $0$ & $\tfrac12$ \\
  \textbf{Action} $\TypedOptimalAction{t}$ & \multirow{2}{*}{$\tfrac12 \TypedForecast{t}{\TypedOptimalAction{t}}{1}$} & \multirow{2}{*}{$\cdots$} & \multirow{2}{*}{$\tfrac12 \TypedForecast{t}{\TypedOptimalAction{t}}{m}$} & \multirow{2}{*}{$\epsilon_t$} & \multirow{2}{*}{$\frac12 - \epsilon_t$} \\ Cost $\cost{\TypedOptimalAction{t}}$ \\
  $\vdots$ & $\vdots$ & $\ddots$ & $\vdots$ & $0$ & $\tfrac12$ \\
  \textbf{Action} $n$ & \multirow{2}{*}{$\tfrac12 \TypedForecast{t}{n}{1}$} & \multirow{2}{*}{$\cdots$} & \multirow{2}{*}{$\tfrac12 \TypedForecast{t}{n}{m}$} & \multirow{2}{*}{$0$} & \multirow{2}{*}{$\tfrac12$} \\ Cost $\cost{n}$ \\
  \bottomrule
\end{tabular}
\caption{Resulting principal-agent problem $\PrincipalAgentProblemPrime$ from the nonlinearity disparity reduction. For each type $t \in [T]$, we use outcome $m + t$ as a way for single contracts to extract the entire welfare. Outcome $m + T + 1$ is serves to rebalance the total probability mass.}
\label{tab:nonlinearity-disparity}
\end{table}
  
  Note that we have correctly guaranteed that the new forecast tensor is made of valid probability distributions, because for any type $t \in [T]$ and action $i \in [n]$, the relevant row sums to one:
  \begin{align*}
    \sum_{j=1}^{m+2} \forecastprime{t}{i}{j}
      = \underbrace{\left[ \sum_{j=1}^m \frac{1}{2} \forecastprime{t}{i}{j} \right]}_{\tfrac12}
      + \underbrace{\left[ \forecastprime{t}{i}{m+1} + \forecastprime{t}{i}{m+2} \right]}_{\frac12}
      = 1
  \end{align*}
  
  This completes our construction. We now argue that the optimal single contract is able to completely capture the new welfare $\OWelfare \PrincipalAgentProblemPrime$. Consider the contract which transfers $\left(1 / \epsilon\right)$ on outcome $(m+1)$: $\ContractVector = \left( 1 / \epsilon \right) \StandardBasisVector{m+1}$. Under this contract, an agent that picks action $i$ is expected to be transferred $\cost{\TypedOptimalAction{t}}$ if $i = \TypedOptimalAction{t}$ and nothing otherwise. Since costs are nonnegative, this means that $\TypedOptimalAction{t}$ is among the actions that maximize the utility of a type $t$ agent. Furthermore, since the principal only transfers the bare minimum, i.e. $\cost{\TypedOptimalAction{t}}$, they keep $\TypedExpectedReward{t}{i^*(t)} - \cost{i^*(t)}$. This is by definition the maximum welfare of a type $t$ agent.
  \begin{align*}
    \OSingle \PrincipalAgentProblemPrime &= \OWelfare \PrincipalAgentProblemPrime
  \end{align*}
  
  By construction, it is easy to see that the welfare has not changed. We scaled down all probabilities by one-half but scaled up all rewards by two, so expected rewards are the same. Since costs have remained constant, the welfare of each action is the same. Hence we get one of our desired statements:
  \begin{align*}
    \OWelfare \PrincipalAgentProblemPrime &= \OWelfare \PrincipalAgentProblem \\
    \OSingle \PrincipalAgentProblemPrime &= \OWelfare \PrincipalAgentProblem
  \end{align*}

  It remains to argue that linear contracts do not gain any profit. Since our two new outcomes have zero reward, they cannot be utilized by linear contracts. Hence, linear contracts are facing a rescaled version of the original problem, and the linear contract with transfer coefficient $\TransferCoefficient \in [0, 1]$ achieves the same profit in both $\PrincipalAgentProblemPrime$ as it did in $\PrincipalAgentProblem$.
  \begin{align*}
    \OLinear \PrincipalAgentProblemPrime &= \OLinear \PrincipalAgentProblem
  \end{align*}
  This completes the proof.
\end{proof}

Applying \Cref{lem:nonlinearity-disparity} to \Cref{thm:main-lower} yields the following corollary.

\begin{corollary}
\label{cor:nonlinearity-disparity}
For all $n, T > 0$, there exists a principal-agent problem $\PrincipalAgentProblem$ with $T$ types, $n+1$ actions, and four outcomes that satisfies
\begin{align*}
  \OLinear \PrincipalAgentProblem \le O\left( \frac{1}{n \log T} \right) \OSingle \PrincipalAgentProblem
  \text{.}
\end{align*}
\end{corollary}

\subsection{Reduction: Information is Power}
\label{subsec:info-is-power}

In this subsection, we give a transformation for (typed) principal-agent problems that enable a type-aware principal to extract the entire welfare as profit, while principals subject to adverse selection cannot extract much more profit than they could in the base problem. For this transformation, we will need the following theoretical tool, capturing the notion that very small additive perturbations to the agent's utility cannot drastically alter the profit.

\begin{definition}[Cost-Stability]
\label{def:cost-stability}
  We say that a principal-agent benchmark $\Measure$ is cost-stable if for every principal-agent problem $\PrincipalAgentProblem$ there exist constants $s \ge 0$, $\tau > 0$ such that for all cost perturbations $\CostVector' \in \mathbb{R}^n, \dist{\CostVector}{\CostVector'}{\infty} \le \tau$, we have that
  \begin{align*}
    \Measure(\CostVector', \ForecastTensor, \RewardVector) \le \Measure \PrincipalAgentProblem + s \dist{\CostVector}{\CostVector'}{\infty}
  \end{align*}
\end{definition}

Cost-stability, when applied to the $\OSingle$, is related to \dutting{} et al.'s notion of $\delta$-incentive compatible contracts \cite{dutting2020complexity}. In \Cref{apx:stability}, we prove the following result, namely that all five of our benchmarks are cost-stable. However, this subsection's reduction only uses the cost-stability of $\OMenu$ and $\OSingle$. 

\begin{restatable}{corollary}{coststablecorollary}
\label{cor:cost-stable}
  The following principal-agent measures are all cost-stable: $\OWelfare$, $\OTypeAware$, $\OMenu$, $\OSingle$, and $\OLinear$.
\end{restatable}

We now present the main theorem of this subsection. We give a proof sketch here and defer the full proof to \Cref{apx:info-is-power}.

\begin{restatable}{theorem}{infoispowerthm}
\label{thm:info-is-power}
  Suppose we have a principal-agent problem $\PrincipalAgentProblem$ with $T$ types, $n$ actions, and $m$ outcomes and a constant $\zeta > 0$. We can construct another principal-agent problem $\PrincipalAgentProblemPrime$ with $T+1$ types, $n$ actions, and $m+2$ outcomes such that:
  \begin{align*}
    \OTypeAware \PrincipalAgentProblemPrime &= \tfrac{T}{T+1} \OWelfare \PrincipalAgentProblem \\
    \OMenu \PrincipalAgentProblemPrime      &\le \tfrac{T}{T+1} \OMenu \PrincipalAgentProblem + \zeta \\
    \OSingle \PrincipalAgentProblemPrime    &\le \tfrac{T}{T+1} \OSingle \PrincipalAgentProblem + \zeta \\
    \OLinear \PrincipalAgentProblemPrime    &= \tfrac{T}{T+1} \OLinear \PrincipalAgentProblem
  \end{align*}
\end{restatable}

\begin{proof}[Proof Sketch]
  We use a parameter $\epsilon \in (0, 1)$ to represent the probability mass dedicated to our two new outcomes, which is chosen in the full proof to keep probabilities valid. One of the outcomes serves to allow type-aware principals to extract the entire welfare as profit, while the other outcome balances probability mass. We introduce a new type to punish a principal unaware of the type who tries to use either of these new outcomes. Doing so is merely expensive for them, not impossible, and so we apply \Cref{cor:cost-stable} to show that barely using these two outcomes is not very powerful.
  
  We give the formal definition of the resulting principal-agent problem $\PrincipalAgentProblemPrime$ here, but defer proving it has the desired properties to \Cref{apx:info-is-power}.
  \begin{alignat*}
    \forall i \in [n] \quad
    &&\Cost{i}' &\triangleq \Cost{i} \\
    \forall t \in [T], i \in [n], j \in [m+2] \quad
    &&\forecastprime{t}{i}{j} &\triangleq \begin{cases}
      (1 - \epsilon) \TypedForecast{t}{i}{j}               & \text{ if } j \in [m] \\
      \epsilon \indicator{i = \TypedOptimalAction{t}}      & \text{ if } j = m + 1 \\
      \epsilon \indicator{i \not = \TypedOptimalAction{t}} & \text{ if } j = m + 2
    \end{cases} \\
    \forall i \in [n], j \in [m+2] \quad
    &&\forecastprime{T+1}{i}{j} &\triangleq \begin{cases}
      0       & \text{ if } j \in [m] \\
      \frac12 & \text{ otherwise}
    \end{cases} \\
    \forall j \in [m+2] \quad
    &&\Reward{j}' &\triangleq \begin{cases}
      \frac{\Reward{j}}{1 - \epsilon} & \text{ if } j \in [m] \\
      0                               & \text{ if } j \in \{m+1, m+2\} \\
    \end{cases}
  \end{alignat*}
\end{proof}

Applying \Cref{thm:info-is-power} to \Cref{cor:two-outcome-linear} yields the following corollary.

\begin{corollary}
\label{cor:info-is-power}
For all $n, T > 0$, there exists a principal-agent problem $\PrincipalAgentProblem$ with $T+1$ types, $n+1$ actions, and four outcomes that satisfies
\begin{align*}
  \OMenu \PrincipalAgentProblem \le O\left( \frac{1}{n \log T} \right) \OTypeAware \PrincipalAgentProblem
  \text{.}
\end{align*}
\end{corollary}

\bibliographystyle{unsrtnat}
\bibliography{contracts-bibliography}

\appendix

\section{Simplifying Assumptions for the Principal-Agent Problem}
\label{subsec:simplifying}

In Section \ref{sec:prelims} we remarked that without loss of generality, we can assume that $r_{1} = c_{1} = 0$; that is, there is always a zero-cost action and a zero-reward outcome. Here we justify these claims.

\begin{itemize}
  \item \textbf{Zero-cost action}: To see why we can assume $c_{1} = 0$, assume $0 < c_{1} \leq c_{2} \leq \dots \leq c_{n}$. We then claim that if we subtract $c_{1}$ from each of these costs (setting $c'_{i} = c_{i} - c_{1}$), we obtain an equivalent principal-agent problem: in particular, if a contract $\Contract{}$ induced action $i$ in the original contract, it continues to induce action $i$ under costs $\CostVector'$. To see this, it suffices to note that since $c'_{i} - c_{i}$ is constant
    
  $$\argmax_i \left( \sum_{j=1}^m \TypedForecast{t}{i}{j} \Contract{j} \right) - \Cost{i} = \argmax_i \left( \sum_{j=1}^m \TypedForecast{t}{i}{j} \Contract{j} \right) - \Cost{i}' \text{,}$$

  and therefore $\OptimalAction(\ContractVector)$ is the same under $\CostVector$ and $\CostVector'$. Since our benchmarks do not depend on cost beyond the induced action of the agent\footnote{Technically, this is untrue as $\OWelfare$ is currently defined -- when $c_1 \neq 0$, we should define welfare as $\OWelfare = \max_{i}(R_{i} - \Delta_i))$, where $\Delta_i = c_{i} - \min_{i'}c_{i'}$ is the minimum payment required to convince the agent to play item $i$. Alternatively, increasing $c_1$ only decreases (our current definition of) welfare, and only makes approximation ratios smaller.}, all our benchmarks and approximation ratios are unchanged under $\CostVector'$. 
    
  \item \textbf{Zero-reward outcome}: To see why we can assume $r_{1} = 0$, assume $0 < r_{1} \leq r_{2} \leq \dots \leq r_{n}$. Again, subtract $r_{1}$ from each of these rewards to obtain a new principal-agent problem with rewards $\RewardVector'$ (given by $r'_{i} = r_{i} - r_{1}$). Now, note that subtracting some amount $r$ from the reward of each outcome results in each of our benchmarks decreasing by $r$ (the expected profit for any mechanism for the principal will decrease by $r$ if all rewards decrease by $r$). Since $\frac{M_{1} - r}{M_{2} - r} > \frac{M_{1}}{M_{2}}$ when $M_{1} > M_{2}$, this only increases our approximation gaps and hence we can assume this without loss of generality.\footnote{Technically, this subtraction does not preserve what contracts are linear, and we would need to consider affine contracts in the general case to correspond with linear contracts in the zero-reward outcome case.}
\end{itemize}
\section{Computational Results}
\label{sec:computation}

The fact that agents have types raises some fundamental questions: given a principal-agent problem $\PrincipalAgentProblem$, can we compute the optimal menu of contracts / single contract? How efficiently can this be done in terms of the size of the input?

In this appendix, we show that computing the optimal contract and optimal menu are \ApxHard{} when the number of types $T$ and number of outcomes $m$ can grow without bound. We subsequently show that there exists an algorithm computes the optimal contract in $O(\min(\poly(n^T,m), \poly((n^2T)^m)$ time. 

\subsection{Hardness Results}

\apxhardness*

\begin{proof}
To show that the problem is \ApxHard, we will reduce from the problem \DSet.

\begin{table}
\centering
\begin{tabular}{cccccc}
  \toprule
  \multirow{2}{*}{Type $t \in [N]$} & \textbf{Outcome} $1$ & \multirow{2}{*}{$\cdots$} & \textbf{Outcome} $j \in [N]$ & \multirow{2}{*}{$\cdots$} & \textbf{Outcome} $N+1$ \\
                                    & Reward $1$           &                           & Reward $1$                   &                           & Reward $0$             \\ \midrule
  \textbf{Action} $1$ & \multirow{2}{*}{$\indicator{\neighbor{1}{t} = 1}$} & \multirow{2}{*}{$\cdots$} & \multirow{2}{*}{$\indicator{\neighbor{1}{t} = j}$} & \multirow{2}{*}{$\cdots$} & \multirow{2}{*}{$0$} \\ Cost $\tfrac12$ \\
  \textbf{Action} $2$ & \multirow{2}{*}{$\indicator{\neighbor{2}{t} = 1}$} & \multirow{2}{*}{$\cdots$} & \multirow{2}{*}{$\indicator{\neighbor{2}{t} = j}$} & \multirow{2}{*}{$\cdots$} & \multirow{2}{*}{$0$} \\ Cost $\tfrac12$ \\
  \textbf{Action} $3$ & \multirow{2}{*}{$\indicator{\neighbor{3}{t} = 1}$} & \multirow{2}{*}{$\cdots$} & \multirow{2}{*}{$\indicator{\neighbor{3}{t} = j}$} & \multirow{2}{*}{$\cdots$} & \multirow{2}{*}{$0$} \\ Cost $\tfrac12$ \\
  \textbf{Action} $4$ & \multirow{2}{*}{$\indicator{t = 1}$} & \multirow{2}{*}{$\cdots$} & \multirow{2}{*}{$\indicator{t = j}$} & \multirow{2}{*}{$\cdots$} & \multirow{2}{*}{$0$} \\ Cost $\tfrac12$ \\
  \textbf{Action} $5$ & \multirow{2}{*}{$0$} & \multirow{2}{*}{$\cdots$} & \multirow{2}{*}{$0$} & \multirow{2}{*}{$\cdots$} & \multirow{2}{*}{$1$} \\ Cost $0$ \\ \midrule
  \multirow{2}{*}{Type $t \in [N+1, 2N]$} & \textbf{Outcome} $1$ & \multirow{2}{*}{$\cdots$} & \textbf{Outcome} $j \in [N]$ & \multirow{2}{*}{$\cdots$} & \textbf{Outcome} $N+1$ \\
                                          & Reward $1$           &                           & Reward $1$                   &                           & Reward $0$             \\ \midrule
  \textbf{Action} $1$ & \multirow{2}{*}{$0$} & \multirow{2}{*}{$\cdots$} & \multirow{2}{*}{$0$} & \multirow{2}{*}{$\cdots$} & \multirow{2}{*}{$1$} \\ Cost $\tfrac12$ \\
  \textbf{Action} $2$ & \multirow{2}{*}{$0$} & \multirow{2}{*}{$\cdots$} & \multirow{2}{*}{$0$} & \multirow{2}{*}{$\cdots$} & \multirow{2}{*}{$1$} \\ Cost $\tfrac12$ \\
  \textbf{Action} $3$ & \multirow{2}{*}{$0$} & \multirow{2}{*}{$\cdots$} & \multirow{2}{*}{$0$} & \multirow{2}{*}{$\cdots$} & \multirow{2}{*}{$1$} \\ Cost $\tfrac12$ \\
  \textbf{Action} $4$ & \multirow{2}{*}{$0$} & \multirow{2}{*}{$\cdots$} & \multirow{2}{*}{$0$} & \multirow{2}{*}{$\cdots$} & \multirow{2}{*}{$1$} \\ Cost $\tfrac12$ \\
  \textbf{Action} $5$ & \multirow{2}{*}{$\indicator{t = N + 1}$} & \multirow{2}{*}{$\cdots$} & \multirow{2}{*}{$\indicator{t = N + j}$} & \multirow{2}{*}{$\cdots$} & \multirow{2}{*}{$0$} \\ Cost $0$ \\ \midrule
  \bottomrule
\end{tabular}
\caption{Counterexample principal-agent problem $\PrincipalAgentProblem$ constructed from an instance of \DSet{}. The types $t \in [N]$ generate one-half unit of profit if their corresponding vertex $t$ is dominated, and the types $t \in [N+1, 2N]$ eliminate one-half unit of profit if their corresponding vertex $(t - N)$ is chosen to be in the dominating set.}
\label{tab:apx-hard}
\end{table}

\textbf{Construction:} Given an instance of \DSet on a graph $G = (V,E)$ with $N$ vertices with maximum degree at most $3$, we will create a principal-agent problem $\PrincipalAgentProblem$ with $T = 2N$ types, $n = 5$ actions, and $m = N + 1$ outcomes. Without loss of generality, we will order the neighbors of $i \in V$ to be non-decreasing and refer to the $k^{\text{th}}$ neighbor as $\neighbor{k}{i}$. Outcome $N + 1$ will be a ``null'' outcome which yields zero reward: $\Reward{N+1} \triangleq 0$. For the remaining outcomes $j \in [N]$, the rewards are one: $\Reward{j} \triangleq 1$. Furthermore, the cost of the first four actions $i \in [4]$ is $\Cost{i} = \tfrac12$ and the final ``null'' action has zero cost: $\Cost{5} = 0$.

For the first $N$ types $t \in [N]$, the final null action always yields the null outcome, i.e. $\TypedForecast{t}{5}{N+1} \triangleq 1$. Action four always yields outcome $t$, i.e. $\TypedForecast{t}{4}{t} \triangleq 1$. The remaining actions $i \in [3]$ always yield outcome $\neighbor{i}{t}$, i.e. $\TypedForecast{t}{i}{\neighbor{i}{t}} \triangleq 1$. Obviously, the remaining probability entries for these first $N$ types are zeros so we get valid probability distributions.

For the last $N$ types $t \in [N+1, 2N]$, actions $i \in [4]$ always produce the null outcome $\TypedForecast{t}{i}{N+1} \triangleq 1$ while the null action $i = 5$ produces outcome $t - N$, i.e. $\TypedForecast{t}{5}{t-N} \triangleq 1$. Again, we fill in the remaining probability entries for these last $N$ types with zeros to get valid probability distributions.

\textbf{Analysis}
($\Leftarrow$)  Given any dominating set $S$ for $G$, consider the contract $\ContractVector$ where $\Contract{t} = \frac12$ for all $t \in S$, and $\Contract{t} = 0$ otherwise. By the definition of dominating set, for each vertex $v \in [N]$, either $v \in S$ or $v$ has some neighbor $w \in S$. Therefore the corresponding type $t \in [N]$ can either take action four or some action $i \in [3]$ to get a transfer of one-half, paying off its cost. The principal can pocket the remaining one-half reward. For the remaining types $t \in [N, 2N]$ each type will produce a deterministic reward of one but the principal will transfer one-half to each $t + i$ such that $i \in S$. Thus the expected profit of this contract is $\frac34 - \frac{|S|}{4N}$. 

($\Rightarrow$) 
Without loss of generality, we assume that the optimal menu of contracts $\ContractMatrix^*$ and optimal single contract $\ContractVector^*$ satisfy the following two properties.
\begin{enumerate}
  \item \label{assum1} They only transfers either zero or one-half on every outcome, i.e. $\TypedContract{t}{j} \in \set{0, \frac12}$ or $\Contract{j} \in \set{0, \frac12}$ for all $j \in [N+1]$.  
  \item \label{assum2} All types $t \in [N]$ choose non-null actions.
\end{enumerate}

If \cref{assum1} is violated, then we can round down all transfer amounts to either zero or one-half. This keeps agent preferences the same but saves money, so it is only (weakly) more profitable.

We next prove \cref{assum2} using \cref{assum1}. Suppose some type $t \in [N]$ chooses the null action $i = N + 1$. Then we know that for all outcomes $j \in \{ v \mid (t,v) \in E \text{ or } v = t \} $, the menu of contracts / single contract offers no transfer. Consider the menu of contracts / single contract which always transfers one-half on outcome $t$, but is otherwise the same. Observe that this new contract gains at least $\frac{1}{4T}$ in expected profit overall by covering an additional type $t \in [T]$ and loses $\frac{1}{4T}$ in expected profit to some $t \in [T, 2T]$. This reduces the number of types $t \in [N]$ that choose the null action by at least one, so applying this a finite number of times yields the desired assumption without any loss in profit.

Consider an optimal menu of contracts / single contract satisfying \cref{assum1} and \cref{assum2} $x^*$. Let $S$ be the set of outcomes where any contract in the menu / the single contract transfers one-half. By \cref{assum2}, $S$ is a dominating set since an agent of type $t \in [N]$ only chooses a non-null action if one of the outcomes in $\{t, \neighbor{1}{t}, \neighbor{2}{t}, \neighbor{3}{t}\}$ has a transfer of one-half. The profit generated by this menu of contracts / single contract is exactly $\frac34 -  \frac{|S|}{4N}$.

Using the work of \cite{chlebik2008approximation}[Proof of Theorem 6], we know that there exist instances of \DSet containing $N$ vertices where it is \NPHard to distinguish if the optimal dominating set has at least $0.28792798 N$ vertices or at most $0.28719007142 N$ vertices. It is hence $\NPHard$ to distinguish if the optimal menu of contracts / single contract extracts expected profit at least $0.67820248214$ or at most $0.678018005$. This completes the proof.
\end{proof}

\subsection{Algorithms to compute the optimal contracts}

\expalg*

\begin{proof}
First we will show an algorithm that computes the optimal contract in time $\poly(n^T,m)$ time. 
The optimal contract will induce some set of actions for each  type of agent.  We  will enumerate over all possible
action profiles $a \in [n]^T$, and compute the optimal contract to induce those actions. Since there are $n^T$ such profiles, 
we will solve $O(n^T)$ linear programs, one for each action profile. 

Given an action profile $a \in [n]^T$, we can compute the optimal contract via the following linear program:
\begin{align*}
\min  &\sum_o (\sum_i F^i_{a_i,o} ) \cdot x_o  \\
\sum_{o \in O} F^i_{a_i, o} \cdot x_o - c_{a_i} &\geq \sum_{o \in O} F^i_{p,o} \cdot x_o - c_{p} \qquad \forall i \in [T], \forall p \neq a_i, p \in [n]  \tag{P1} \label{const:fair} \\
   x &\geq 0 
\end{align*} 

In  constraint~\ref{const:fair}, we require the contract $x$ forces an agent of type $i$ 
to prefer action $a_i$ over all other $p$. Since we know that there exists an action with $0$ cost, we
know that the optimal contract will either induce this action profile or be infeasible. 
This LP has  $n(T+1)$ constraints, and can be solved in $\poly(n\cdot T \cdot m)$ time. 

Next, we show an algorithm that computes the optimal contract in time $\poly((n^2T)^m)$. 
For each agent of type $i$, and each pair of actions $a_1 , a_2 \in [n]$, an optimal contract will 
induce a preference between $a_1$ and $a_2$. Each of these can be encoded as a 
constraint of the form $\cref{const:fair}$. Since there are $\binom{n}{2}$ possible pairs of constraints and $T$ different types of agents, there can be at most $O(n^2T)$ distinct constraints. 
Since the optimal contract $x^*$ is a solution to an LP determined by constraints of the form~$\cref{const:fair}$, it must be an extreme point
 determined by some set of 
constraints of size $m$ from set of all possible constraints. Therefore, we can enumerate all possible subsets of constraints of size $m$ from $O(n^2T)$ constraints and evaluate the resulting contract. This procedure can be done in time $O(n^2T)^m)$. 
\end{proof}

\section{Proof of the Two Outcome Case}
\label{apx:two-outcome}

In this appendix, we give the proof of \Cref{lem:two-outcome-linear}. We restate the lemma below for convenience.

\twooutcomelemma*

\begin{proof}
  Fix an optimal menu of contracts $\ContractMatrix$:
  \begin{align*}
    \ContractMatrix = \left(
      \TypedContractVector{1} = \left(\TypedContract{1}{1}, \TypedContract{1}{2}\right),
      \ldots,
      \TypedContractVector{T} = \left(\TypedContract{T}{1}, \TypedContract{T}{2}\right),
    \right)
    \text{.}
  \end{align*}
  
  We simply give a single linear contract that will beat the menu. Before we define it, we first define two helper values.
  \begin{align*}
    \bar{x} &\triangleq \max_t \TypedContract{t}{2} - \TypedContract{t}{1} \\
    \tilde{x} &\triangleq \min \{ \max \{ \bar{x}, 0 \}, \reward{2} \}
  \end{align*}
  Informally, $\bar{x}$ is the maximum difference any contract in the menu gave for outcome two over outcome one, and $\tilde{x}$ is the same but capped to the range $[0, \reward{2}]$.
  
  Consider the linear contract $\ContractVector \triangleq \left(0, \tilde{x}\right)$. We will prove it is better via a hybrid argument. In particular, let us run a though experiment where we restrict all agents to their choice of action under the menu $\ContractMatrix$ but replace the menu with this linear contract $\ContractVector$. We claim that under this fixed-action restriction, the linear contract achieves (weakly) more profit than the menu did. To see why, observe that by definition, $\tilde{x}$ either (i) still equals $\bar{x}$, (ii) was capped down to $\reward{2}$, or (iii) was capped up to zero. In cases (i) and (ii), consider a type $t$ agent that chose action $i$. They could have misreported their type to be $\bar{t} \triangleq \argmax_t \TypedContract{t}{2} - \TypedContract{t}{1}$, but did not. By incentive compatibility, we know the following about their transfer:
  \begin{align*}
    \TypedForecast{t}{i}{1} \TypedContract{t}{1}
    + \TypedForecast{t}{i}{2} \TypedContract{t}{2}
    &\stackrel{(IC)}{\ge} \TypedForecast{t}{i}{1} \TypedContract{\bar{t}}{1}
    + \TypedForecast{t}{i}{2} \TypedContract{\bar{t}}{2} \\
    &= \TypedContract{\bar{t}}{1}
    + \TypedForecast{t}{i}{2} \left(
      \TypedContract{\bar{t}}{2} - \TypedContract{\bar{t}}{1}
    \right) \\
    &\ge \TypedForecast{t}{i}{2} \tilde{x}
  \end{align*}
  Hence $\ContractVector$ does not increase their expected transfer. We argue about case (iii) separately since we rounded up, but it is covered by noting that any contract $\TypedContractVector{t}$ trivially transfers (weakly) more than our all-zeros linear contract $\ContractVector$.
  
  We now compare this hybrid state to the plain linear contract setting we actually want to compare with. What happens if we now allow agents to once again change their choice of action? We might be worried that a type $t \in [T]$ agent will decide to deviate from action $i \in [n]$ to action $i' \in [n]$. Suppose that the original contract $\TypedContractVector{t}$ did not have an excessive difference ($\TypedContract{t}{2} - \TypedContract{t}{1} \le \reward{2}$), and this deviation lowers the probability of outcome two ($\TypedForecast{t}{i'}{2} < \TypedForecast{t}{i}{2}$):
  \begin{align*}
    (0) \TypedForecast{t}{i'}{1} + (\tilde{x}) \TypedForecast{t}{i'}{2} - \cost{i'}
      &> (0) \TypedForecast{t}{i}{1} + (\tilde{x}) \TypedForecast{t}{i}{2} - \cost{i} \\
    \underbrace{(\tilde{x})}_{(\text{at least } \TypedContract{t}{2} - \TypedContract{t}{1} \text{ by assumption})}
      \underbrace{\left( \TypedForecast{t}{i'}{2} - \TypedForecast{t}{i}{2} \right)}_{(\text{less than } 0 \text{ by assumption})}
      &> \cost{i'} - \cost{i} \\
    \left(\TypedContract{t}{2} - \TypedContract{t}{1}\right)
      \left( \TypedForecast{t}{i'}{2} - \TypedForecast{t}{i}{2} \right)
      &> \cost{i'} - \cost{i} \\
    \left(\TypedContract{t}{2} - \TypedContract{t}{1}\right) \TypedForecast{t}{i'}{2} - \cost{i'}
      &> \left(\TypedContract{t}{2} - \TypedContract{t}{1}\right) \TypedForecast{t}{i}{2} - \cost{i} \\
    \TypedContract{t}{1} \TypedForecast{t}{i'}{1} + \TypedContract{t}{2} \TypedForecast{t}{i'}{2} - \cost{i'}
      &> \TypedContract{t}{1} \TypedForecast{t}{i}{1} + \TypedContract{t}{2} \TypedForecast{t}{i}{2} - \cost{i}
  \end{align*}
  We have concluded that this deviation to action $i'$ should have occurred under the original menu $\TypedContractVector{t}$, which contradicts its incentive compatibility. Hence, any type $t$ whose original contract did not have an excessive difference ($\TypedContract{t}{2} - \TypedContract{t}{1} \le \reward{2}$) can only deviate to (weakly) increase the probability of outcome two. Such deviations only (weakly) increase our profit, namely by $(\reward{2} - \tilde{x}) \left(\TypedForecast{t}{i'}{2} - \TypedForecast{t}{i}{2}\right)$. On the other hand, if type $t$'s original contract did have an excessive difference, i.e. $\TypedContract{t}{2} - \TypedContract{t}{1} \le \reward{2}$, then the menu could only lose money on this type. Since our linear contract never loses money on any type, we do at least as well. Hence across all types $t \in [T]$, this linear contract makes (weakly) more profit when the fixed-action restriction is lifted.
  
  Since we have already proved that the linear contract under the fixed-action restriction makes more profit than the menu did, we have shown that a linear contract makes as much profit as the optimal menu, as desired. This completes the proof.
\end{proof}
\section{Proof of the Information is Power Reduction}
\label{apx:info-is-power}

In this appendix, we give the proof of \Cref{thm:info-is-power}. We restate the theorem below for convenience.

\infoispowerthm*

\begin{proof}
  Our goal is to permit a type-aware contract to extract the full welfare as profit but not allow menus of contracts, single contracts, or linear contracts to extract much additional profit. The overall proof plan bears similarities to that of \Cref{lem:nonlinearity-disparity}, but we need to make an effort to prevent menus of contracts and single contracts from using our new outcomes. This is done via the introduction of a new type, which serves to severely punish any principal who tries to use a new outcome without being type-aware.
  
  Again we need a parameter $\epsilon \in (0, 1)$, which we choose later in the proof. It now represents how much probability mass we want to use for our two new outcomes (which previously was always one-half). We plan to scale down the probability of existing outcomes by $(1 - \epsilon)$ and scale up the rewards by $\frac{1}{1 - \epsilon}$. We again define $\TypedOptimalAction{t}$ to be the welfare-maximizing action for a type $t$ agent in the original principal-agent problem $\PrincipalAgentProblem$, breaking ties arbitrarily.
  
  Combining the above ideas, we formally define our new principal-agent problem $\PrincipalAgentProblemPrime$ as follows (see \Cref{tab:info-is-power} for a depiction of this construction in tabular form).
  \begin{alignat*}
    \forall i \in [n] \quad
    &&\Cost{i}' &\triangleq \Cost{i} \\
    \forall t \in [T], i \in [n], j \in [m+2] \quad
    &&\forecastprime{t}{i}{j} &\triangleq \begin{cases}
      (1 - \epsilon) \TypedForecast{t}{i}{j}               & \text{ if } j \in [m] \\
      \epsilon \indicator{i = \TypedOptimalAction{t}}      & \text{ if } j = m + 1 \\
      \epsilon \indicator{i \not = \TypedOptimalAction{t}} & \text{ if } j = m + 2
    \end{cases} \\
    \forall i \in [n], j \in [m+2] \quad
    &&\forecastprime{T+1}{i}{j} &\triangleq \begin{cases}
      0       & \text{ if } j \in [m] \\
      \frac12 & \text{ otherwise}
    \end{cases} \\
    \forall j \in [m+2] \quad
    &&\Reward{j}' &\triangleq \begin{cases}
      \frac{\Reward{j}}{1 - \epsilon} & \text{ if } j \in [m] \\
      0                               & \text{ if } j \in \{m+1, m+2\} \\
    \end{cases}
  \end{alignat*}
  
  \begin{table}
\centering
\begin{tabular}{cccccc}
  \toprule
  \multirow{2}{*}{Type $t \in [T]$} & \textbf{Outcome} $1$                        & \multirow{2}{*}{$\cdots$} & \textbf{Outcome} $m$                        & \textbf{Outcome} $m+1$ & \textbf{Outcome} $m+2$ \\
                                    & Reward $\tfrac{1}{1 - \epsilon} \reward{1}$ &                           & Reward $\tfrac{1}{1 - \epsilon} \reward{m}$ & Reward $0$             & Reward $0$             \\ \midrule
  \textbf{Action} $1$ & \multirow{2}{*}{$(1 - \epsilon) \TypedForecast{t}{1}{1}$} & \multirow{2}{*}{$\cdots$} & \multirow{2}{*}{$(1 - \epsilon) \TypedForecast{t}{1}{m}$} & \multirow{2}{*}{$0$} & \multirow{2}{*}{$\epsilon$} \\ Cost $\cost{1}$ \\
  \textbf{Action} $2$ & \multirow{2}{*}{$(1 - \epsilon) \TypedForecast{t}{2}{1}$} & \multirow{2}{*}{$\cdots$} & \multirow{2}{*}{$(1 - \epsilon) \TypedForecast{t}{2}{m}$} & \multirow{2}{*}{$0$} & \multirow{2}{*}{$\epsilon$} \\ Cost $\cost{2}$ \\
  $\vdots$ & $\vdots$ & $\ddots$ & $\vdots$ & $0$ & $\epsilon$ \\
  \textbf{Action} $\TypedOptimalAction{t}$ & \multirow{2}{*}{$(1 - \epsilon) \TypedForecast{t}{\TypedOptimalAction{t}}{1}$} & \multirow{2}{*}{$\cdots$} & \multirow{2}{*}{$(1 - \epsilon) \TypedForecast{t}{\TypedOptimalAction{t}}{m}$} & \multirow{2}{*}{$\epsilon$} & \multirow{2}{*}{$0$} \\ Cost $\cost{i^*(t)}$ \\
  $\vdots$ & $\vdots$ & $\ddots$ & $\vdots$ & $0$ & $\epsilon$ \\
  \textbf{Action} $n$ & \multirow{2}{*}{$(1 - \epsilon) \TypedForecast{t}{n}{1}$} & \multirow{2}{*}{$\cdots$} & \multirow{2}{*}{$(1 - \epsilon) \TypedForecast{t}{n}{m}$} & \multirow{2}{*}{$0$} & \multirow{2}{*}{$\epsilon$} \\ Cost $\cost{n}$ \\ \midrule
  \multirow{2}{*}{Type $t = T + 1$} & \textbf{Outcome} $1$                        & \multirow{2}{*}{$\cdots$} & \textbf{Outcome} $m$                        & \textbf{Outcome} $m+1$ & \textbf{Outcome} $m+2$ \\
                                    & Reward $\tfrac{1}{1 - \epsilon} \reward{1}$ &                           & Reward $\tfrac{1}{1 - \epsilon} \reward{m}$ & Reward $0$             & Reward $0$             \\ \midrule
  \textbf{Action} $1$ & \multirow{2}{*}{$0$} & \multirow{2}{*}{$\cdots$} & \multirow{2}{*}{$0$} & \multirow{2}{*}{$\tfrac12$} & \multirow{2}{*}{$\tfrac12$} \\ Cost $\cost{1}$ \\
  $\vdots$ & $\vdots$ & $\ddots$ & $\vdots$ & $\tfrac12$ & $\tfrac12$ \\
  \textbf{Action} $n$ & \multirow{2}{*}{$0$} & \multirow{2}{*}{$\cdots$} & \multirow{2}{*}{$0$} & \multirow{2}{*}{$\tfrac12$} & \multirow{2}{*}{$\tfrac12$} \\ Cost $\cost{n}$ \\
  \bottomrule
\end{tabular}
\caption{Resulting principal-agent problem $\PrincipalAgentProblemPrime$ from the information is power reduction. For each old type $t \in [T]$, we use outcome $m + 1$ as a way for the full information principal to extract the entire welfare as revenue. Outcome $m + 2$ serves the rebalance the total probability mass. The new type $t = T + 1$ serves to prevent uninformed principals from utilizing these additional outcomes.}
\label{tab:info-is-power}
\end{table}
  
  Note that we have correctly guaranteed that the new forecast tensor has valid probability distributions, because for any old type $t \in [T]$ and action $i \in [n]$, the relevant row sums to one:
  \begin{align*}
    \sum_{j=1}^{m+2} \forecastprime{t}{i}{j}
      = \underbrace{\left[ \sum_{j=1}^m (1 - \epsilon) \forecastprime{t}{i}{j} \right]}_{1 - \epsilon}
      + \underbrace{\left[ \forecastprime{t}{i}{m+1} + \forecastprime{t}{i}{m+2} \right]}_{\epsilon}
      = 1
  \end{align*}
  Of course for the new type $t = T + 1$, rows trivially sum to one.
  
  This completes our construction, and we are now ready to argue that a type-aware principal is able to completely capture the new welfare $\OWelfare \PrincipalAgentProblemPrime$. Consider the following strategy for a type-aware principal. If the agent has an old type $t \in [T]$, then the principal offers the contract that transfers $\left(\cost{\TypedOptimalAction{t}} / \epsilon \right)$ on outcome $(m + 1)$: $\TypedContractVector{t} = \left( \cost{\TypedOptimalAction{t}} / \epsilon \right) \StandardBasisVector{m+1}$. Under this contract, an agent that picks action $i$ is expected to be transferred $\cost{\TypedOptimalAction{t}}$ if $i = \TypedOptimalAction{t}$ and nothing otherwise. Since costs are nonnegative, this means that action $\TypedOptimalAction{t}$ is among the actions that maximize the utility of a type $t$ agent. Furthermore, single the principal managed to transfer the bare minimum, i.e. $\cost{\TypedOptimalAction{t}}$, they keep $\TypedExpectedReward{t}{\TypedOptimalAction{t}} - \cost{\TypedOptimalAction{t}}$. This is by definition the maximum welfare of a type $t$ agent. If the agent has the new type $t = T + 1$, then the principal has nothing to gain and hence offers the all-zeroes contract: $\TypedContractVector{t} = \vec{0}$. For every type, this strategy lets the principal capture all the welfare:
  \begin{align*}
    \OSingle \PrincipalAgentProblemPrime &= \OWelfare \PrincipalAgentProblemPrime
  \end{align*}
  
  By construction, the (expected) welfare has dropped by a factor of $\tfrac{T}{T+1}$, because we diluted the type pool with a zero-welfare type and in the base problem we scaled probabilities down but scaled up rewards to match. Hence we get one of our desired statements:
  \begin{align*}
    \OWelfare \PrincipalAgentProblemPrime &= \tfrac{T}{T+1} \OWelfare \PrincipalAgentProblem \\
    \OTypeAware \PrincipalAgentProblemPrime &= \tfrac{T}{T+1} \OWelfare \PrincipalAgentProblem
  \end{align*}
  
  It remains to argue that linear contracts, single contracts, and menus of contracts do not gain much profit. We begin with the easiest case, linear contracts. Since our two new outcomes have zero reward, they cannot be utilized by linear contracts. Hence, linear contracts are facing a rescaled version of the original problem, and the linear contract with transfer coefficient $\TransferCoefficient \in [0, 1]$ achieves the same profit in both $\PrincipalAgentProblemPrime$ as it did in $\PrincipalAgentProblem$, modulo the addition of a zero-welfare type to the pool of types.
  \begin{align*}
    \OLinear \PrincipalAgentProblemPrime &= \tfrac{T}{T+1} \OLinear \PrincipalAgentProblem
  \end{align*}
  
  We now argue about single contracts and menus of contracts. Our goal is to prevent these from utilizing our two new outcomes using the new type $(T + 1)$. Suppose that a single contract or any contract in a menu of contracts offers one of the following:
  \begin{align*}
    \contract{m+1} &\ge (2T) \OWelfare \PrincipalAgentProblem\text{ or} \\
    \contract{m+2} &\ge (2T) \OWelfare \PrincipalAgentProblem\text{.}
  \end{align*}
  Such a contract or menu of contracts incurs a cost of $(T) \OWelfare \PrincipalAgentProblem$ and makes nothing on the new type $(T + 1)$ and therefore cannot attain strictly positive expected profit over all types. Hence this would at best tie the single all-zeroes contract $\ContractVector = \vec{0}$. We have shown that without loss of generality, we can assume that the optimal single contract / all contracts in the optimal menu of contracts satisfy:
  \begin{align*}
    \contract{m+1} &\le (2T) \OWelfare \PrincipalAgentProblem\text{ and} \\
    \contract{m+2} &\le (2T) \OWelfare \PrincipalAgentProblem\text{.}
  \end{align*}
  
  Such a single contract / menu of contracts does not have much power to influence the agent via these two outcomes. Compared to a single contract / menu of contracts whose transfers are restricted to the old types $t \in [T]$, such a single contract / menu of contracts has at most an additive $\pm (2T) \epsilon \OWelfare \PrincipalAgentProblem$ effect on the agent's utility for any action. In other words, this is at most as powerful as working in the original principal-agent problem $\PrincipalAgentProblem$ and being able to perturb the costs by at most $\epsilon (2T) \OWelfare \PrincipalAgentProblem$ (modulo an additional zero-welfare type diluting the type pool). Our technical result \Cref{cor:cost-stable} lets us control the effect of this additive perturbation: there exists constants\footnote{Technically the proof-generated constants might be different for $\OSingle$ and $\OMenu$, but we can pick the larger error scaling term $s$ and the smaller perturbation threshold $\tau$.} $s \ge 0, \tau > 0$ such that as long as $\epsilon (2T) \OWelfare \PrincipalAgentProblem \le \tau$,
  \begin{align*}
    \tfrac{T+1}{T} \OSingle \PrincipalAgentProblemPrime &\le \OSingle \PrincipalAgentProblem + s \epsilon (2T) \OWelfare \PrincipalAgentProblem\text{ and} \\
    \tfrac{T+1}{T} \OMenu \PrincipalAgentProblem &\le \OMenu \PrincipalAgentProblem + s \epsilon (2T) \OWelfare \PrincipalAgentProblem\text{.}
  \end{align*}
  We now choose the parameter $\epsilon$ to (i) be less than $\tfrac12$, (ii) satisfy the error perturbation threshold, and (iii) result in at most $\zeta$ additive error:
  \begin{align*}
    \epsilon \triangleq \min \left\{\frac12, \frac{\min\{\tau, \zeta\}}{(2T) \OWelfare \PrincipalAgentProblem} \right\}\text{.}
  \end{align*}
  Plugging in this choice of $\epsilon$ implies:
  \begin{align*}
    \OSingle \PrincipalAgentProblemPrime &\le \tfrac{T}{T+1} \OSingle \PrincipalAgentProblem + \zeta \text{ and} \\
    \OMenu \PrincipalAgentProblemPrime &\le \tfrac{T}{T+1} \OMenu \PrincipalAgentProblem + \zeta \text{.}
  \end{align*}
  This completes the proof.
\end{proof}
\section{Cost-Stability}
\label{apx:stability}

In this appendix, we prove \Cref{cor:cost-stable}, one benchmark at a time. At a high level, the intuition is that we can set an appropriate perturbation threshold $\tau$ to prevent infeasible mappings of types to actions from becoming feasible, and we can set our error scaling term $s$ to account for the growth of the feasible regions of already-feasible mappings. We begin with linear contracts, for which we explicitly derive good choices for $\tau$ and $s$. We then proceed to the more complex single contracts and menus of contracts, for which our argument is similar but involves reasoning about linear programs. We finish with the simple type-aware and welfare cases.

\begin{lemma}\label{lem:linear-stable}
  The expected profit of the best linear contract, $\OLinear$, is cost-stable.
\end{lemma}

\begin{proof}
  We begin with some useful definitions for this proof. The error term $\epsilon$ defined to be the maximum amount any cost is perturbed: $\epsilon = \dist{c}{c'}{\infty}$. We use $u^{(t)}_i(\TransferCoefficient)$ to denote the utility of a type $t \in [T]$ agent choosing action $i \in [n]$ under a linear contract of transfer coefficient $\TransferCoefficient \in [0, 1]$ in the original principal-agent problem $\PrincipalAgentProblem$; we use $u'^{(t)}_i(\TransferCoefficient)$ to denote the same for the perturbed principal-agent problem $(\CostVector', \ForecastTensor, \RewardVector)$. We remove the $i$ subscripts to denote the max over all actions; $u^{(t)}(\TransferCoefficient)$ denotes the maximum utility our type $t$ agent can achieve under a linear contract of transfer coefficient $\TransferCoefficient \in [0, 1]$ in the original problem $\PrincipalAgentProblem$ and $u'^{(t)}(\TransferCoefficient)$ denotes the same for the perturbed problem $(\CostVector', \ForecastTensor, \RewardVector)$. Formally, we can write these four functions as:
  \begin{align*}
    u^{(t)}_i(\TransferCoefficient) &= \TransferCoefficient R^{(t)}_i - c_i \\
    u'^{(t)}_i(\TransferCoefficient) &= \TransferCoefficient R^{(t)}_i - c'_i = u_i(\TransferCoefficient) \pm \epsilon \\
    u^{(t)}(\TransferCoefficient) &= \max_i u^{(t)}_i(\TransferCoefficient) \\
    u'^{(t)}(\TransferCoefficient) &= \max_i u'^{(t)}_i (\TransferCoefficient) = u^{(t)}(\TransferCoefficient) \pm \epsilon
  \end{align*}
  
  In the original principal-agent problem $\PrincipalAgentProblem$, the principal could convince a type $t$ agent to choose action $i \in [n]$ if said action was utility-maximizing, i.e. $u^{(t)}_i(\TransferCoefficient) = u^{(t)}(\TransferCoefficient)$. Imagine that we fix a particular mapping from types to actions $\actionmapping$, signifying that we want type $t$ to choose action $i_t$. As we vary our choice of $\TransferCoefficient$, this mapping is either (i) can never be chosen or (ii) can be chosen when $\TransferCoefficient$ is in some closed subset of $[0, 1]$. We will denote the set of $\actionmapping$ that fall into the former class as $I_1$ and the set that fall into the latter class as $I_2$. The proof plan is that when the principal adjusts all costs (effectively adjusting all utilities) by $\epsilon$, the region of $\TransferCoefficient$ that produce a specific class (i) mapping remains empty while the region that produces a specific class (ii) mapping does not grow much larger.
  
  \textbf{Case 1.} The mapping $\optactionmapping \in I_1$ is currently infeasible for all $\TransferCoefficient \in [0, 1]$, and we need ensure that we preserve this property when we choose our perturbation threshold $\tau$. To do so, we examine the largest amount any action $i^*_t$ lags behind the best action: $\max_{t \in T} \left[ u^{(t)}(\TransferCoefficient) - u^{(t)}_{i^*_t}(\TransferCoefficient) \right]$; we know that for all $\TransferCoefficient \in [0, 1]$, this expression is strictly positive. We will choose our perturbation threshold $\tau$ to guarantee that the most lagging action cannot catch up:
  \begin{align*}
    \tau &\triangleq \frac13 \left[ \min_{\actionmapping \in I_1, \TransferCoefficient \in [0, 1]} \max_{t \in T} u^{(t)}(\TransferCoefficient) - u^{(t)}_{i_t}(\TransferCoefficient) \right] > 0
  \end{align*}
  
  We now need to prove that this choice was good enough to get our desired guarantee. Consider the perturbed problem $(\CostVector', \ForecastTensor, \RewardVector)$ where the analogous lag expression is $u'^{(t)}(\TransferCoefficient) - u'^{(t)}_{i^*_t}(\TransferCoefficient)$. Observe that for all $\TransferCoefficient \in [0, 1]$:
  \begin{align*}
    &\phantom{=} \max_{t \in T} \left[ u'^{(t)}(\TransferCoefficient) - u'^{(t)}_{i^*_t}(\TransferCoefficient) \right] \\
    &= \max_{t \in T} \left[ u^{(t)}(\TransferCoefficient) - u^{(t)}_{i^*_t}(\TransferCoefficient) \right] \pm 2 \epsilon \\
    &\ge \max_{t \in T} \left[ u^{(t)}(\TransferCoefficient) - u^{(t)}_{i^*_t}(\TransferCoefficient) \right] - 2 \tau \\
    &\ge \max_{t \in T} \left[ u^{(t)}(\TransferCoefficient) - u^{(t)}_{i^*_t}(\TransferCoefficient) \right]
         - \frac23 \left[ \min_{\actionmapping \in I_1, \tilde{\TransferCoefficient} \in [0, 1]}
                          \max_{t \in T} \left[ u^{(t)}(\tilde{\TransferCoefficient}) - u^{(t)}_{i_t}(\tilde{\TransferCoefficient}) \right] \right] \\
    &\ge \max_{t \in T} \left[ u^{(t)}(\TransferCoefficient) - u^{(t)}_{i^*_t}(\TransferCoefficient) \right] - \frac23 \max_{t \in T} \left[ u^{(t)}(\TransferCoefficient) - u^{(t)}_{i^*}(\TransferCoefficient) \right] \\
    &\ge \frac13 \max_{t \in T} \left[ u^{(t)}(\TransferCoefficient) - u^{(t)}_{i^*_t}(\TransferCoefficient) \right] \\
    &> 0
  \end{align*}
  Hence we have successfully guaranteed that the mapping $\optactionmapping \in I_1$ remains infeasible via our choice of perturbation threshold $\tau$.
  
  \textbf{Case 2.} The mapping $\optactionmapping \in I_2$ is currently feasible, and we want to control the growth of its feasible region. If we are indeed committed to forcing the mapping $\optactionmapping$, then we would always prefer a smaller sharing coefficient $\TransferCoefficient$ (since that maximizes resulting profit). In other words, we just need to control how much the smallest feasible $\TransferCoefficient$ decreases. Suppose that in the original problem $\PrincipalAgentProblem$ the smallest feasible value was $\TransferCoefficient^*$. Formally, we know that:
  \begin{align*}
    u^{(t)}(\TransferCoefficient^*) &= u^{(t)}_{i^*_t}(\TransferCoefficient^*) \\
    \forall \TransferCoefficient < \TransferCoefficient^* \quad
    u^{(t)}(\TransferCoefficient) &> u^{(t)}_{i^*_t}(\TransferCoefficient) \\
    \forall \TransferCoefficient < \TransferCoefficient^* \quad
    \exists i \in [n] \quad
    u^{(t)}_i(\TransferCoefficient) &> u^{(t)}_{i^*_t}(\TransferCoefficient)
  \end{align*}
  
  But our agent's utility for any action is linear in $\TransferCoefficient$. In other words, the functions $u^{(t)}_i(\TransferCoefficient)$ and $u^{(t)}_{i^*_t}(\TransferCoefficient)$ can only cross once, and some particular action $i \in [n]$ witnesses the optimality of $\TransferCoefficient^*$.
  \begin{align}
    \label{ineq:original}
    \exists i \in [n] \quad
    \forall \TransferCoefficient < \TransferCoefficient^* \quad
    u^{(t)}_i(\TransferCoefficient) &> u^{(t)}_{i^*_t}(\TransferCoefficient)
  \end{align}
  
  To determine how much the smallest feasible $\TransferCoefficient$ could decrease in the perturbed problem, we recall the need for our desired actions to not lag utility-wise (i.e. be tied for the best utility). In other words, we know that:
  \begin{align}
    \nonumber
    u'^{(t)}(\TransferCoefficient) &=
    u'^{(t)}_{i^*_t}(\TransferCoefficient) \\
    \nonumber
    \implies \forall i \in [n] \quad
    u'^{(t)}_i(\TransferCoefficient) &\le
    u'^{(t)}_{i^*_t}(\TransferCoefficient) \\
    \label{ineq:perturbed}
    \implies \forall i \in [n] \quad
    u^{(t)}_i(\TransferCoefficient) - 2 \epsilon &\le
    u^{(t)}_{i^*_t}(\TransferCoefficient)
  \end{align}
  Inequality~\ref{ineq:perturbed} will help us control how much lower $\alpha$ can be set in the perturbed problem, and hence how much additional profit the principal can extract there. We can instantiate it for any action $i \in [n]$, and we will choose the action defined by Inequality~\ref{ineq:original}.
  
  We want to deduce a small fact about our action $i$, so we split into two subcases, depending on the relationship between the expected principal rewards $R^{(t)}_{i^*_t}$ (expected reward when type $t$ plays desired action $i^*_t$) and $R^{(t)}_{i}$ (expected reward when type $t$ plays alternate action $i$).
  
  \textbf{Subcase 2(a).} Suppose our type-action pair satisfies $R^{(t)}_{i^*_t} \le R^{(t)}_i$, so we can deduce:
  \begin{align*}
    u^{(t)}_{i^*_t}(\TransferCoefficient^*)          &\ge u^{(t)}_i(\TransferCoefficient^*) & \\
    \TransferCoefficient^* R^{(t)}_{i^*_t} - c_{i^*} &\ge \TransferCoefficient^* R^{(t)}_i - c_i & \\
    \TransferCoefficient   R^{(t)}_{i^*_t} - c_{i^*} &\ge \TransferCoefficient   R^{(t)}_i - c_i & \forall \TransferCoefficient \le \TransferCoefficient^* \\
    u^{(t)}_{i^*_t}(\TransferCoefficient)            &\ge u^{(t)}_i(\TransferCoefficient) & \forall \TransferCoefficient \le \TransferCoefficient^*
  \end{align*}
  We have shown such an alternate action $i$ cannot be the action guaranteed by Inequality~\ref{ineq:original}. Hence that action must fall into the subcase below instead.
  
  \textbf{Subcase 2(b).} This is the only subcase left, so we have deduced that the action $i$ guaranteed by Inequality~\ref{ineq:original} satisfies $R^{(t)}_{i^*_t} > R^{(t)}_i$. Combining what we know yields:
  \begin{align*}
  \begin{array}{rll}
    u^{(t)}_{i^*_t}(\TransferCoefficient)
      &< u^{(t)}_i(\TransferCoefficient) & \forall \TransferCoefficient < \TransferCoefficient^* \\
    u^{(t)}_{i^*_t}(\TransferCoefficient)
      &< u^{(t)}_i(\TransferCoefficient) + \frac{2\epsilon}{R^{(t)}_{i^*_t} - R^{(t)}_i} \left(R^{(t)}_{i^*_t} - R^{(t)}_i\right) - 2\epsilon & \forall \TransferCoefficient < \TransferCoefficient^* \\
    u^{(t)}_{i^*_t}(\TransferCoefficient) - \frac{2\epsilon}{R^{(t)}_{i^*_t} - R^{(t)}_i} R^{(t)}_{i^*_t}
      &< u^{(t)}_i(\TransferCoefficient) - \frac{2\epsilon}{R^{(t)}_{i^*_t} - R^{(t)}_i} R^{(t)}_i - 2\epsilon & \forall \TransferCoefficient < \TransferCoefficient^* \\
    u^{(t)}_{i^*_t}\left( \TransferCoefficient - \frac{2\epsilon}{R^{(t)}_{i^*_t} - R^{(t)}_i} \right)
      &< u^{(t)}_i\left( \TransferCoefficient  - \frac{2\epsilon}{R^{(t)}_{i^*_t} - R^{(t)}_i} \right) - 2\epsilon & \forall \TransferCoefficient < \TransferCoefficient^* \\
    u^{(t)}_{i^*_t}(\TransferCoefficient)& < u^{(t)}_i(\TransferCoefficient) - 2 \epsilon & \forall \TransferCoefficient < \TransferCoefficient^* - \frac{2 \epsilon}{R^{(t)}_{i^*_t} - R^{(t)}_i}
  \end{array}
  \end{align*}
  
  We have proven that the smallest feasible $\TransferCoefficient$ may decrease by at most $\frac{2 \epsilon}{R^{(t)}_{i^*_t} - R^{(t)}_i}$ (note: there was also an unhandled case where $\TransferCoefficient^*$ was already zero, but that leads to this conclusion as well). This increases the potential profit of $\optactionmapping$ by at most $\frac{2 \epsilon}{R^{(t)}_{i^*_t} - R^{(t)}_i} \sum_{\tilde{t}=1}^T R^{(t)}_{i^*_t}$. We can cover this by choosing our error scaling term $s$ to be:
  \begin{align*}
    s &\triangleq \max_{\actionmapping \in I_2, t \in [T], i \in [n] . R^{(t)}_{i^*_t} > R^{(t)}_i} \frac{2 \sum_{\tilde{t}} R^{(\tilde{t})}_{i^*_t}}{R^{(t)}_{i^*_t} - R^{(t)}_i}
  \end{align*}
  (or zero, if for all Case 2 mappings $\TransferCoefficient^* = 0$). This completes the proof.
\end{proof}

We are now ready to take a look at the more complicated case of single contracts.

\begin{lemma}\label{lem:single-stable}
  The expected profit of the best single contract, $\OSingle$, is cost-stable.
\end{lemma}

\begin{proof}
  As in the proof of Lemma~\ref{lem:linear-stable}, we define the error term $\epsilon = \dist{c}{c'}{\infty}$. This proof proceeds in a similar fashion, except that our arguments are now based on linear programs and we are less precise in choosing $s$ and $\tau$.
  
  In the original principal-agent problem $\PrincipalAgentProblem$, the principal could get the mapping $\actionmapping$ with minimum expected transfers by solving the following linear program:
  \begin{align*}
  LP_{\actionmapping} \left\{
    \begin{array}{rll}
                  \min_x & \sum_{t=1}^T \sum_{j=1}^m F^{(t)}_{i_t j} x_j \\
      \\
      \text{subject to } & \sum_{j=1}^m \left[ F^{(t)}_{i_tj} - F^{(t)}_{ij} \right] x_j \ge c_{i_t} - c_i & \forall t \in [T], i \in [n] \\
                         & x_j \ge 0                                                                       & \forall j \in [m]
    \end{array}
  \right.
  \end{align*}
  Let $I_1$ be the set of mappings $\actionmapping$ whose linear program $LP_{\actionmapping}$ is infeasible; $I_2$, feasible.
  
  \textbf{Case 1.} The mapping $\optactionmapping$ currently has an infeasible LP and we want to maintain this status quo. We will choose our perturbation threshold $\tau$ to keep it so. We can determine how far this mapping is from feasible (again, we pretend the principal can pick a different perturbation for each type, giving it more power) by writing another linear program:
  \begin{align*}
  LP'_{\actionmapping} \left\{
    \begin{array}{rll}
      \min_{x, \epsilon} & \epsilon \\ 
      \\
      \text{subject to } & \sum_{j=1}^m \left[ F^{(t)}_{i_tj} - F^{(t)}_{ij} \right] x_j + 2 \epsilon \ge c_{i_t} - c_i & \forall t \in [T], i \in [n] \\
                         & x_j \ge 0                                                                                    & \forall j \in [m]
    \end{array}
  \right.
  \end{align*}
  Since $LP_{\optactionmapping}$ was infeasible, we know that $LP'_{\optactionmapping}$ has a strictly positive optimal value (the latter is feasible because very large perturbations e.g. larger than all costs suffice). We hence can choose our perturbation threshold $\tau$ to be:
  \begin{align*}
    \tau \triangleq \frac12 \min_{\actionmapping \in I_1} \text{value}(LP'_{\actionmapping})
  \end{align*}
  
  \textbf{Case 2.} The mapping $\optactionmapping$ currently has a feasible LP and we want to control how much additional profit the principal can extract from it. Consider the following dual to $LP_{\actionmapping}$:
  \begin{align*}
  D_{\actionmapping} \left\{
    \begin{array}{rll}
                  \max_y & \sum_{t=1}^T \sum_{i=1}^n (c_{i_t} - c_i) y^{(t)}_i \\
      \\
      \text{subject to } & \sum_{t=1}^T \sum_{i=1}^n \left[ F^{(t)}_{i_tj} - F^{(t)}_{ij} \right] y_i \le \sum_{t=1}^T F^{(t)}_{i_tj} & \forall j \in [m] \\
                         & y^{(t)}_i \ge 0                                                                                            & \forall t \in [T], i \in [n]
    \end{array}
  \right.
  \end{align*}
  
  By strong LP-duality, we know that $LP_{\optactionmapping}$ and $D_{\optactionmapping}$ have the same optimal value. Additionally, due to \cite{grotschel2012geometric} there exists a dual point $y^*$ which achieves this optimal value whose $\ell_\infty$ norm is exponential in the bit-complexity of $D_{\optactionmapping}$. Now, suppose that the costs $c$ get perturbed by $\epsilon$ into $c'$. Our previous dual solution $y^*$ is still feasible, because only the dual objective gets changed. In fact, its dual objective value drops by at most the $\ell_\infty$ norm of $y^*$. Hence even under perturbed costs $c'$ the optimal value of the dual (and primal, by strong LP-duality) also drops by that amount. Profit goes up by at most the same amount, and hence it suffices to choose our error scaling term $s$ to be:
  \begin{align*}
    s \triangleq \max_{\actionmapping \in I_2} 2^{1 + O(\text{bit-complexity}(D_{\actionmapping})}
  \end{align*}
  
  This completes the proof.
\end{proof}

The proof for menus of contracts just involves more complicated linear programs.

\begin{lemma}\label{lem:menu-stable}
  The expected profit of the best menu of contracts, $\OMenu$, is cost-stable.
\end{lemma}

\begin{proof}
  The proof is identical to that of Lemma~\ref{lem:single-stable}, but with the alternate linear programs:
  \begin{align*}
  LP_{\actionmapping} &\left\{
    \begin{array}{rll}
                  \min_x & \sum_{t=1}^T \sum_{j=1}^m F^{(t)}_{i_t j} x^{(t)}_j \\
      \\
      \text{subject to } & \sum_{j=1}^m \left[ F^{(t)}_{i_tj} x^{(t)}_j - F^{(t)}_{ij} x^{(t')}_j \right] \ge c_{i_t} - c_i & \forall t,t' \in [T], i \in [n] \\
                         & x^{(t)}_j \ge 0                                                                                  & \forall t \in [T], j \in [m]
    \end{array}
  \right. \\ \cline{1-2}
  LP'_{\actionmapping} &\left\{
    \begin{array}{rll}
      \min_{x, \epsilon} & \epsilon \\
      \\
      \text{subject to } & \sum_{j=1}^m \left[ F^{(t)}_{i_tj} x^{(t)}_j - F^{(t)}_{ij} x^{(t')}_j \right] + 2 \epsilon \ge c_{i_t} - c_i & \forall t,t' \in [T], i \in [n] \\
                         & x^{(t)}_j \ge 0                                                                                               & \forall t \in [T], j \in [m]
    \end{array}
  \right. \\ \cline{1-2}
  D_{\actionmapping} &\left\{
    \begin{array}{rll}
                  \max_y & \sum_{t=1}^T \sum_{t'=1}^T \sum_{i=1}^n (c_{i_t} - c_i) y^{(t, t')}_i \\
      \\
      \text{subject to } & \sum_{t'=1}^T \sum_{i=1}^n \left[ F^{(t)}_{i_tj} y^{(t, t')}_i - F^{(t)}_{ij} y^{(t', t)}_i \right] \le \sum_{t=1}^T F^{(t)}_{i_tj} & \forall t \in [T], j \in [m] \\
                         & y^{(t, t')}_i \ge 0                                                                                                  & \forall t, t' \in [T], i \in [n]
    \end{array}
  \right.
  \end{align*}
  
  The same choices (plugging in these new linear programs, of course) of perturbation threshold $\tau$ and error scaling term $s$ give the desired properties. This completes the proof.
\end{proof}

\begin{corollary}
\label{cor:type-aware-stable}
  The expected profit of the best contract per type, $\OTypeAware$, is cost-stable.
\end{corollary}

\begin{proof}
  Given any principal-agent problem $\PrincipalAgentProblem$, we can just break it down by type and apply \Cref{lem:single-stable} to each type. We choose our perturbation threshold to be the minimum of the resulting perturbation thresholds, and our error scaling $s$ to be the maximum of the resulting error scalings.
\end{proof}

\begin{lemma}
\label{lem:welfare-stable}
  The maximum expected welfare, $\OWelfare$, is cost-stable.
\end{lemma}

\begin{proof}
  From the definition of $\OWelfare$, it is clear that changing all costs by an additive $\epsilon$ results in $\OWelfare$ increasing by at most $\epsilon$ as well. Hence we can arbitrarily pick our perturbation threshold $\tau \triangleq 1$ and we should pick our error scaling to be $s \triangleq 1$.
\end{proof}

Combining \Cref{lem:linear-stable}, \Cref{lem:single-stable}, \Cref{lem:menu-stable}, \Cref{cor:type-aware-stable}, and \Cref{lem:welfare-stable} gives us the goal of this section, \Cref{cor:cost-stable}, which is restated below for convenience.

\coststablecorollary*
\section{Gap Instance separating Optimal Menu and Optimal Single Contract}
\label{apx:gap3v4}

\begin{lemma}
\label{lem:gap3v4}
We show that there exists a family of instances with $n=3$ actions, where the ratio of the $\OMenu \text{ vs }\OSingle$  is at least $2 - \epsilon $ for any constant $\epsilon>0$. 
\end{lemma}
\begin{proof}

\textbf{Construction.}
For every value of $k \in \N$ we will construct a typed principal-agent problem $(\costvector,\ForecastTensor,\rewardvector)$ 
with $T = k+1$ and $n=3$ and $m=k+1$. Once again, we designate
outcome $k+1$ to be the null outcome which produces reward 
$r_{k+1}=0$. The remaining outcomes $o \in [k]$ produce reward
$r_o = 1$. Similarly there is a null action $1$, which always 
produces the null outcome $k+1$, i.e.~$\TypedForecast{t}{1}{k+1}=1$ for all $t \in [T]$. The 
remaining actions  have cost $\cost{2} = C'$ where $C' = (1-\frac1k)$ and 
$\cost{3} = \frac{C'}{k} =  (\frac1k - \frac{1}{k^2})$. 

We will also assign the types different weights which are polynomial in $k$. This is 
without loss of generality as we can simply make them uniform by creating duplicate copies of 
each type. Let $\alpha = \frac1k$.  We will weight the first $k$ types to occur with probability 
$ \frac{(1-\alpha)}{k} = \frac{k-1}{k^2}$ and the last type to occur with probability $\alpha = \frac1k$. 

For types $t \in [k+1]$, we will set
$\TypedForecast{t}{2}{t}= 1.0$.  Observe that the 
last type $k+1$ produces the null outcome on action $2$.   For types $t \in [k]$, we will set 
$\TypedForecast{t}{3}{k+1}= 1.0$. 
For the  remaining type $k+1$,  we set $\TypedForecast{k+1}{3}{j} = \frac1k$ for $j \in [k]$. 
Unlike the previous settings we will set the first $k$ types to appear with probability 
$\frac{k-1}{k^2}$ and the last type to appear with probability $\frac{1}{k}$.  We can easily remove 
this assumption by simply making many copies of each distribution.

\textbf{Optimal Menu.}
Consider the menu consisting of $k$ contracts where the $i^{th}$ contract transfers on the $C'$ on the 
outcome $i$. This contract extracts revenue  
$(1-\alpha)\cdot (1-C') + \alpha \cdot (1-C'/k)$ which is $\frac2k - \frac{2}{k^2} + \frac{1}{k^3}$. 

\textbf{Optimal Single Contract.}
We will show that the best single contract will take be one of two types. 
Either the contract transfers $C'$  on the first $p$ outcomes for $p \leq k$ or it transfers $C$ on the first $k$ outcomes. 
The first contract achieves  a revenue of $\frac{p}{k} ( 1- \frac1k) (1-C')+ \frac1k(1- \frac{p}{k} C')$ which approaches $\frac1k$ as we
maximize $p \leq k$.   On the other hand, we can transfer $C$ on all the $k$ outcomes. This contract will achieve a revenue of $\frac1k \cdot \frac1k \cdot (1-C)$ As we set $k \to \infty$, the ratio between the best single contract approaches $2$. 

To argue that the optimal single contract must take one of the two forms, let us consider 
an optimal contract $x^*$. First observe 
for type $k+1$, since all outcomes have equal probability, a contract's ability to convince 
agent $k+1$ to take action $3$ is only affected by whether $\sum_{i \in [k]} x^*_i \geq C$. 

Suppose $x^*$  convinces $p\geq 1$ of the first $k$ types to choose a non-null 
action (i.e.~action $2$). Consider the contract $\hat{x}_i = C'$ for $i \leq p$. 
By construction, the new contract will convince at least $p$ of the first $k$ types. 
Since $C = C'/k$, we know that this will also convince type $k+1$.  It is easy to see that $\hat{x} \leq x^*$ and therefore the principal will only make more revenue in the new contract. Therefore this is the first type of contract, we considered above. 

Suppose $x^*$ does not convince any agent in the first $[k]$ types but only convinces the last agent. 
In this case, we can consider the contract $\hat{x}_i = C$ for all $i \in [k]$. Clearly, this contract convinces the agent and its sum is the minimum possible transfer that one can do to convince agent of type $k+1$.  This is the second type of contract we consider above. 
\end{proof}
\section{Types with Unequal Costs} \label{apx:nonuniform-costs}

Throughout this paper, we have made the assumption that different types $t \in [T]$ differ in their forecast probabilities $\TypedForecastMatrix{t}$ but not in the costs of their actions $\CostVector$. While in some settings this decision is justified (for example, whenever actions are parameterized by cost), it is natural to ask what happens when costs for actions can vary between types (or even more generally, when types have completely different sets of actions). In this appendix, we consider a \emph{cost-varying} typed principal-agent problem with the modification that each type $t \in [T]$ has its own cost vector $\TypedCostVector{t}$. We abuse notation in this appendix by using $\CostVector$ to denote the collection of these typed cost vectors $\left\{ \TypedCostVector{t} \right\}_{t=1}^T$.

In this appendix, we show that this generalization transforms all of our $O(n \log T)$ and $\Omega(n \log T)$ bounds into $O(nT)$ and $\Omega(nT)$ bounds, respectively. That is, the consistency of costs between different types is essential to the logarithmic dependence on $T$ that we prove in Theorem~\ref{thm:main-lower}.

To do this, we show that for every cost-varying typed principal-agent problem with $n$ actions and $T$ types, we can find a standard (i.e. untyped) principal-agent problem with $(n-1)T+1$ actions (and a single type) with the same value under our $\OLinear$ and $\OWelfare$ benchmarks. 

\begin{lemma}\label{lem:diff_costs_upper}
For every typed principal-agent problem $\PrincipalAgentProblem$ with $n$ actions and $T$ types, there exists a standard principal-agent problem $\PrincipalAgentProblemPrime$ with $(n-1)T+1$ actions (and a single type) such that
\begin{eqnarray*}
\OLinear \PrincipalAgentProblem &=& \OLinear \PrincipalAgentProblemPrime\\
\OWelfare \PrincipalAgentProblem &=& \OWelfare \PrincipalAgentProblemPrime.
\end{eqnarray*}
\end{lemma}
\begin{proof}
The main intuition is that both $\OLinear$ and $\OWelfare$ can be computed solely in terms of the piecewise-linear function $U(\alpha)$, the expected utility of an agent when offered a linear contract with ratio $\alpha$ (see the discussion in Section \ref{sec:linear}). For a cost-varying typed principal-agent problem with $n$ actions and $T$ types, this piecewise-linear function has $(n-1)T$ breakpoints; we will therefore show how to construct a principal-agent problem with $(n-1)T+1$ actions and one type with the same function $U(\alpha)$.

By following the proof of Theorem \ref{thm:main-upper}, we can find values $(\alpha_i, W_i, D_i)$ for $i \in [(n-1)T]$ such that $\alpha_i$ is increasing in $i$ and such that $U(\alpha) = W_i\alpha - D_i$ for $\alpha \in [\alpha_i, \alpha_{i+1})$. Now, consider the following standard principal-agent problem $\PrincipalAgentProblemPrime$ with $(n-1)T+1$ actions, labelled $0$ through $(n-1)T$. There are two outcomes, one with reward $0$ (the ``null outcome'') and one with reward $1$ (the ``reward outcome''). Action $0$ has cost $0$ and probability $0$ on the reward outcome. For each $i \in [(n-1)T]$, action $i$ has cost $D_i$ and probability $W_i$ on the reward outcome. (If any of the $W_i > 1$, increase the reward of the reward outcome to $\max_i W_i$ and scale down all the probabilities by $\max_{i} W_i$). 

Let $U'(\alpha)$ be the expected utility of an agent in the problem $\PrincipalAgentProblemPrime$ when offered a linear contract with ratio $\alpha$. It is straightforward to verify that by construction, $U'(\alpha) = U(\alpha)$. The lemma follows.


\end{proof}

Vice versa, given any standard principal-agent problem with $(n-1)T + 1$ actions, we can transform it into a typed contract with $n$ actions and $T$ types with the same value under these two benchmarks.

\begin{lemma}\label{lem:diff_costs_lower}
For every standard principal-agent problem $\PrincipalAgentProblemPrime$ with $(n-1)T+1$ actions, there exists a cost-varying typed principal-agent problem $\PrincipalAgentProblem$ with $n$ actions and $T$ types, such that
\begin{eqnarray*}
\OLinear \PrincipalAgentProblem &=& \OLinear \PrincipalAgentProblemPrime\\
\OWelfare \PrincipalAgentProblem &=& \OWelfare \PrincipalAgentProblemPrime.
\end{eqnarray*}
\end{lemma}
\begin{proof}
As in the proof of Lemma \ref{lem:diff_costs_upper}, we will construct a cost-varying typed principal-agent problem with $n$ actions and $T$ types with the same utility function $U(\alpha)$ as that of our original untyped principal-agent problem. 

Specifically, let $U'(\alpha)$ be the expected utility of the agent when offered a linear contract with ratio $\alpha$ in the problem $\PrincipalAgentProblemPrime$. $U'(\alpha)$ is piecewise-linear, with $(n-1)T + 1$ pieces and $(n-1)T$ breakpoints. As before, define $(\alpha_i, W_i, D_i)$ for $i \in [(n-1)T]$ such that $U(\alpha) = W_{i}\alpha - D_i$ for $\alpha \in [\alpha_i, \alpha_{i+1})$. 

Construct the following typed principal-agent problem $\PrincipalAgentProblem$. As before, there are two outcomes, a null outcome with reward 0 and a reward outcome with reward $1$ (to be scaled later). Each type $t$ (for $t \in [T]$) will have the following $n$ actions, labelled $0$ through $n-1$. Action $0$ has cost $0$ and reward $0$. For each $i \in [n-1]$, action $i$ has cost $D_{n(t-1) + i} - D_{n(t-1)}$ and probability $W_{n(t-1) + i} - W_{n(t-1)}$ on the reward outcome (scaling down probabilities while scaling up the reward when necessary). By following the proof of Theorem \ref{thm:main-upper}, it is straightforward to check that this leads to a $U(\alpha)$ equal to $U'(\alpha)$. 
\end{proof}

Now, recall that the gap between $\OLinear$ and $\OWelfare$ for the standard principal-agent problem is $\Theta(n)$ (as shown in \cite{dutting2019simple}). The fact that this gap is tight along with Lemmas \ref{lem:diff_costs_upper} and \ref{lem:diff_costs_lower} implies the following corollary.

\begin{corollary}\label{cor:ntgap}
Let $\PrincipalAgentProblem$ be a cost-varying typed principal-agent problem. Then we have that

$$\OLinear \PrincipalAgentProblem \leq O\left(\frac{1}{nT}\right)\OWelfare\PrincipalAgentProblem.$$

Moreover, this is tight: there exists a cost-varying typed principal-agent problem $\PrincipalAgentProblem$ where

$$\OLinear \PrincipalAgentProblem \geq \Omega\left(\frac{1}{nT}\right)\OWelfare\PrincipalAgentProblem.$$
\end{corollary}

This establishes the $\Theta(nT)$ gap between $\OLinear$ and $\OWelfare$ in this setting. We note that the proof of Lemma~\ref{lem:two-outcome-linear} argues type-by-type and still holds under cost-varying types. Furthermore, our above reduction preserves the fact that the counterexample has two outcomes, satisfying the conditions of this lemma. Our other two reduction results, Lemma~\ref{lem:nonlinearity-disparity} and Theorem~\ref{thm:info-is-power}, also continue to hold under cost-varying types. This establishes the remaining gaps when applied to Corollary~\ref{cor:ntgap}. 

\end{document}